\pdfoutput=1
\documentclass{article}
\usepackage{graphicx}
\usepackage{amsmath}
\usepackage{amssymb}
\usepackage{amsthm}
\usepackage{latexsym}

\newcommand{\mynewpage}{}

\newtheorem{thm}{Theorem}
\newtheorem{lem}[thm]{Lemma}

\newtheorem{ex}{Example}

\begin{document}

\begin{center}
{\Large Nonexistence of minimal-time solutions 
for some variations of the firing squad synchronization problem 
having simple geometric configurations}
\end{center}

\begin{center}
{\large Kojiro Kobayashi}
\end{center}
\vspace*{-2.0em}
\begin{center}
\texttt{kojiro@gol.com}
\end{center}

\begin{center}
{September 25, 2019}
\end{center}

\medskip

\noindent
{\bf Abstract}

\medskip

We prove nonexistence of minimal-time solutions 
for three variations of the firing squad synchronization problem 
(FSSP, for short).
Configurations of these variations are 
paths in the two-dimensional grid space 
having simple geometric shapes.
In the first variation a configuration is an L-shaped path 
such that the ratio of the length of horizontal line 
to that of the vertical line 
is fixed.  The general may be at any position.
In the second and the third variations a configuration is 
a rectangular wall such that the ratio of the length 
of the two horizontal walls to that of the two vertical walls is 
fixed.  
The general is at the left down corner 
in the second variation and may be at any position 
in the third variation.
We use the idea used in the proof of 
Yamashita et al's recent similar result 
for variations of FSSP with sub-generals.

\medskip


\mynewpage

\section{Introduction}
\label{section:introduction}

\subsection{The original firing squad synchronization problem}
\label{subsection:original_fssp}

The firing squad synchronization problem (FSSP) 
is an interesting puzzle in automata theory.
It was proposed by J. Myhill in 1957 and became 
widely known to researchers 
by an article by E. F. Moore (\cite{Moore}) 
that gave a concise description of the problem.
The problem is to design a finite automaton $A$ 
that satisfies certain conditions.  

$A$ has two input terminals, one from the left and another from the right.
It also has two output terminals, one to the left and another to the right. 
The value of an output terminal of $A$ at a time is 
the state of $A$ at that time.
We connect $n$ copies of $A$ to construct a network of 
finite automata as shown in Figure \ref{figure:fig000}.

\begin{figure}[htb]
\begin{center}
\includegraphics[scale=1.0]{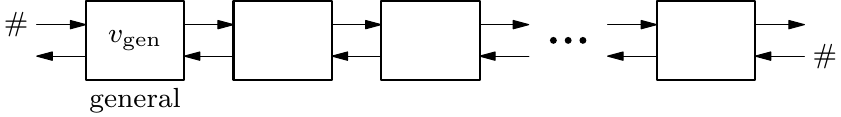}
\end{center}
\caption{A network of $n$ copies of a finite automaton $A$.}
\label{figure:fig000}
\end{figure}

Here $n$ ($\geq 1$) may be any natural number.
The symbol $\#$ is a special signal that means that 
there are no copies there.
We call the copies of $A$ in the network 
the {\it nodes}\/ of the network.
We call the leftmost node of the network its {\it general}\/.

The set of states of $A$ contains 
three different special states: 
the ``general'' state ${\rm G}$, 
the {\it quiescent}\/ state ${\rm Q}$ and 
the {\it firing}\/ state ${\rm F}$.
At time $0$ the general is in the general state ${\rm G}$ 
and all other nodes are in the quiescent state ${\rm Q}$.
Then the state of a node $v$ at a time $t$ ($\geq 0$) 
is uniquely determined by the state transition function of $A$.
We require that if the state of $A$ is ${\rm Q}$ and 
both of its two inputs are either ${\rm Q}$ or $\#$ 
at a time then the state of $A$ at the next time is ${\rm Q}$ again.
Intuitively a node in the quiescent state remains in that state 
until it is activated by an adjacent node that is in 
a nonquiescent state (that is, a state that is not ${\rm Q}$).

The FSSP is the problem to construct $A$ that satisfies the 
following condition: 
for any size $n$ of the network there is a time 
$t_{n}$ (that may depend on $n$) such that any node in the 
network enters the firing state ${\rm F}$ at the time 
$t_{n}$ for the first time.

\mynewpage

The problem is not difficult and it is a good example 
of application of the divide and solve strategy in the 
design of algorithms.
The problem has many variations and they have been 
studied mainly from two standpoints.
One is to find fast solutions, that is, solutions with 
small firing time ($t_{n}$, in the case of the original FSSP).
Another is to find small solutions, that is, solutions with 
small numbers of states.

The research of these problems is practically important 
because it gives distributed computing algorithms 
that synchronize large networks in small time and 
with small memory.
Moreover, it is important because FSSP and its variations 
are mathematical models of one of the most fundamental protocols 
in the design of distributed computing: 
to realize global synchronization using only local 
information exchanges.
By studying these problems we understand the 
theoretical limit of performance of 
algorithms that realize this fundamental protocol.

In this paper we concentrate on obtaining fast solutions for 
variations of FSSP, especially obtaining fastest solutions.
We explain basic notions and state the problems we consider.

\mynewpage

We call a network shown in Figure \ref{figure:fig000} 
a {\it configuration}\/ of FSSP.  For each $n$ ($\geq 1$), 
by $C_{n}$ we denote the configuration having $n$ nodes.
We denote the general of $C_{n}$ (that is, the leftmost node 
of $C_{n}$) by $v_{\rm gen}$.
For any node $v$ in $C_{n}$, any time $t$, and any finite automaton $A$, 
by ${\rm st}(v, t, C_{n}, A)$ we denote the state of $v$ at time $t$ 
under the assumption that all nodes in $C$ 
are copies of $A$.
Then, ${\rm st}(v, 0, C_{n}, A)$ is ${\rm G}$ or ${\rm Q}$ 
according as $v = v_{\rm gen}$ or not.
Moreover, $A$ is a solution if for any $n$ ($\geq 1$) 
there exists a time $t_{n}$ such that 
${\rm st}(v, t, C_{n}, A) \not= {\rm F}$ for any node $v$ 
and any time $t$ such that $0 \leq t < t_{n}$ 
and ${\rm st}(v, t_{n}, C_{n}, A) = {\rm F}$ for any node $v$.
We call the time $t_{n}$ mentioned in the condition 
the {\it firing time}\/ of the solution $A$ for $C_{n}$ and 
denote it by ${\rm ft}(C_{n}, A)$.

If we try to construct a solution of FSSP, we usually 
find a solution $A$ such that 
${\rm ft}(C_{n}, A) = 3n + O(\log n)$.
We are interested in finding faster solutions.
Concerning this, we know a very interesting result that 
there is a fastest solution $\tilde{A}$ 
among all solutions, that is, 
a solution $\tilde{A}$ satisfying the condition:  
\begin{equation}
(\forall A) (\forall C_{n}) [ {\rm ft}(C_{n}, \tilde{A}) 
\leq {\rm ft}(C_{n}, A) ],
\label{equation:eq000}
\end{equation}
where $A$ ranges over all solutions.
We call a solution $\tilde{A}$ that satisfies 
the condition (\ref{equation:eq000}) 
a {\it minimal-time solution}.
Minimal-time solutions were found by E. Goto \cite{Goto_1962}, 
A. Waksman \cite{Waksman_1966} and R. Balzer \cite{Balzer_1967}.
Their solutions also showed that 
the firing time ${\rm ft}(C_{n}, \tilde{A})$ 
of a minimal-time solution $\tilde{A}$ is 
$2n - 2$ for $n \geq 2$ and $1$ for $n = 1$.
It is certain that this discovery of minimal-time solutions 
triggered the long lasting research on FSSP.

\mynewpage

\subsection{Variations of FSSP and minimal-time solutions}
\label{subsection:variations}

Researchers have studied many variations of FSSP and 
tried to find fast and small solutions.
These variations include (one-way or two-way) rings, 
squares, rectangles, cuboids, 
(one-way or two-way) tori, 
Cayley graphs, 
regions in grid spaces, 
paths in grid spaces, directed graphs, and 
undirected graphs.
We refer the readers to 
\cite{
Goldstein_Kobayashi_SIAM_2005, 
Goldstein_Kobayashi_SIAM_2012,
Mazoyer_Survey, 
Napoli_Parente, 
Umeo_2005_Survey}
for surveys on these variations.

For each of these variations 
the problem to know whether it has minimal-time solutions 
or not is one of the most basic problems.
More specifically, for each variation for which we do not know 
minimal-time solutions, it is desirable to show a rigorous 
proof of their nonexistence.
For some of variations we know such proofs.
We show examples.

(1) In \cite{Mazoyer_1995} Mazoyer considered a very simple variation of 
FSSP such that there is only one configuration that consists of two nodes 
and the time for the two nodes to communicate is unbounded.
He proved that this variation has no minimal-time solutions.

(2) In \cite{Schmid_Worsch_IFIP_2004} Schmid and Worsch considered 
the variations of the original FSSP such that more than one 
general is allowed. They showed that for the variation such that 
generals are activated independently in time (the ``asynchronous'' 
case), there are no minimal-time solutions.

(3) In \cite{Yamashita_et_al_2014} Yamashita et al considered 
variations of the original FSSP such that a special state 
called ``sub-general'' is allowed.  For one variation they 
showed nonexistence of minimal-time solutions.
This work is closely related to our present paper and 
we will explain this result in more detail later.

\mynewpage

In this paper we consider the problem to prove 
nonexistence of minimal-time solutions only for variations 
that satisfy the following four conditions:
\begin{enumerate}
\item[(1)]  Each node in a configuration is either 
a general node or a non-general node.
(For example, there are no ``sub-general'' type nodes.)
\item[(2)]  Each configuration has exactly one general 
and it is activated at time $0$.
\item[(3)]  Information exchanges between nodes are synchronous 
and require exactly one unit time.
\item[(4)]  There are no restrictions on the modes of information exchanges 
between nodes. (For example, there is no restriction such as 
``at most one bit.'') 
\end{enumerate}
Intuitively, we consider variations of FSSP only with 
respect to shapes of configurations.
The above mentioned three results are excluded from 
our study.

\mynewpage

As an approach to prove nonexistence of minimal-time solutions 
of variations of FSSP, there has been one which we call 
a ``complexity-theoretical approach.'' 
We will give a brief survey on this approach and its results.
We need to define one notion.

Let $\Gamma$ be a variation of FSSP that has a solution.
For each configuration $C$ of $\Gamma$ we define the value 
\[
{\rm mft}_{\Gamma}(C) = \min \{ {\rm ft}(C, A) ~|~ 
\text{$A$ is a solution of $\Gamma$} \}
\]
and call it {\it the minimum firing time}\/ of $C$.
This value is well-defined because $\Gamma$ has at least 
one solution.
Using this notion we can define a {\it minimal-time solution}\/ 
of $\Gamma$ as a solution $\tilde{A}$ of $\Gamma$ 
such that 
\begin{equation}
(\forall C) [ {\rm ft}(C, \tilde{A}) = {\rm mft}_{\Gamma}(C) ].
\label{equation:eq012}
\end{equation}

Now we have two definitions of {\it minimal-time solutions}\/ 
of $\Gamma$, one using (\ref{equation:eq000}) and another using 
(\ref{equation:eq012}).
However we can easily prove that these two definitions 
are equivalent (\cite{Kobayashi_TCS_2014}).

\mynewpage

The key idea for the complexity-theoretical approach is 
the statement:
\begin{eqnarray*}
\lefteqn{\text{$\Gamma$ has a minimal-time solution}} \\
& \Longrightarrow & \text{${\rm mft}_{\Gamma}(C)$ can be 
computed in polynomial time.}
\end{eqnarray*}
This is true because ${\rm mft}(C)$ is the time when copies of 
a minimal-time solution $\tilde{A}$ placed on $C$ fire and 
this time can be computed in polynomial time by simulating 
behaviors of copies of $\tilde{A}$ placed on $C$.
Next we select a statement $\mathcal{A}$ that we believe 
to be false and prove:
\[
\text{${\rm mft}_{\Gamma}(C)$ can be computed in 
polynomial time} 
\Longrightarrow \mathcal{A}.
\]
If we succeed in proving this, we have proved:
\[
\text{$\Gamma$ has a minimal-time solution} 
\Longrightarrow \mathcal{A}, 
\]
\noindent
and this is a convincing evidence that 
$\Gamma$ has no minimal-time solutions 
for those who believe that $\mathcal{A}$ is false.

\mynewpage

The author and a coauthor obtained four results 
with this approach.
\begin{enumerate}
\item[$\bullet$] If 2PATH 
(the FSSP of paths in the two-dimensional 
grid space) has a minimal-time solution then 
2PEP can be solved in polynomial time 
(\cite{Kobayashi_TCS_2001}).
\item[$\bullet$] If 3PATH 
(the FSSP of paths in the three-dimensional 
grid space) has a minimal-time solution then 
${\rm P} = {\rm NP}$ 
(\cite{Goldstein_Kobayashi_SIAM_2005}).
\item[$\bullet$] If DN (the FSSP of directed networks) 
has a minimal-time solution then 
there is an algorithm that computes diameters 
of directed graphs in $\tilde{O}(eD + nD^{2})$ time 
(\cite{Goldstein_Kobayashi_SIAM_2012}).
\item[$\bullet$] If UN (the FSSP of undirected networks) 
has a minimal-time solution then 
there is an algorithm that computes diameters of 
undirected graphs in $\tilde{O}(eD^{4})$ time 
(\cite{Goldstein_Kobayashi_SIAM_2012}).
\end{enumerate}

Here $\text{\rm 2PEP}$ 
(the ``two-dimensional path extension problem'')
is a problem about paths in the two-dimensional 
grid space.
For it, at present we know only exponential time algorithms.
Although it is in NP we are unable to show its NP-completeness.
The result for 3PATH is a convincing circumstantial evidence for 
nonexistence because $\text{\rm P} = \text{\rm NP}$ 
is widely believed to be false.

As for results on DN and UN, 
$\tilde{O}$ notation is the variation of $O$ notation 
where we ignore $\log$ factors 
(for example, $O(n^{3} (\log n)^{2} \log \log n)$ 
is $\tilde{O}(n^{3})$) and 
$e$, $n$, $D$ are the number of edges, the number of nodes, 
and the diameter of the graph.
The algorithms mentioned in the two results have 
factors $D$, $D^{2}$ or $D^{4}$ in their running times 
and hence have a distinguishing feature that 
they are fast when diameters $D$ of graphs are small.
At present, all known diameter algorithms first solve 
the all-pairs shortest path problem and 
we know no diameter algorithms 
that are fast when diameters of graphs are small.

\mynewpage

\subsection{Yamashita et al's results}
\label{subsection:yamashita_et_al}

Recently a nonexistence proof of minimal-time solutions 
was obtained for one variation using a completely different 
approach.  This proof needs no assumptions from the complexity 
theory (the theory of computational complexity).
We explain this result.

As we mentioned above, 
Yamashita et al (\cite{Yamashita_et_al_2014}) 
introduced a variation of the original FSSP 
that allows a special state ${\rm G}_{\rm SUB}$ 
called a {\it sub-general state}\/.
Although ${\rm G}_{\rm SUB}$ is different from ${\rm Q}$, 
a node in ${\rm G}_{\rm SUB}$ behaves like a node in 
${\rm Q}$.  
A more precise statement of this is as follows.

Let $s$ be a state and $s_{0}$, $s_{1}$ be either states or 
the signal $\#$. 
Let $\delta(s, s_{0}, s_{1})$ denote the state of a node $v$ 
at time $t + 1$ when at time $t$, $v$ is in the state $s$ 
and $v$ receives $s_{0}$ and $s_{1}$ from its left and right 
input terminals respectively.
Then $\delta({\rm G}_{\rm SUB}, s_{0}, s_{1}) = {\rm G}_{\rm SUB}$ 
if both of $s_{0}$, $s_{1}$ are ${\rm Q}$, ${\rm G}_{\rm SUB}$ 
or $\#$.
(A node in the state ${\rm G}_{\rm SUB}$ remains in the same state 
if values of both of its input terminals are 
${\rm Q}$, ${\rm G}_{\rm SUB}$ or $\#$.)
$\delta(s, {\rm Q}, s_{1}) = \delta(s, {\rm G}_{\rm SUB}, s_{1})$ 
and 
$\delta(s, s_{0}, {\rm Q}) = \delta(s, s_{0}, {\rm G}_{\rm SUB})$ 
for any $s$, $s_{0}$, $s_{1}$.
(If $v$ and $v'$ are adjacent, $v$ cannot distinguish 
the two cases ``$v'$ is in ${\rm Q}$'' and 
``$v'$ is in ${\rm G}_{\rm SUB}$.'')

\mynewpage

Let $a$, $b$ be some fixed positive integers.
For $l \geq 1$, let $A_{a+b, b, l}$ be 
the configuration shown in Figure \ref{figure:fig053} (a).
It is a sequence of $l (a + b) + 1$ nodes 
$p_{0}, p_{1}, \ldots, p_{(a+b)l}$.
At time $0$ the leftmost node $p_{0}$ ($v_{\rm gen}$) 
is in ${\rm G}$, the node $p_{al}$ is in ${\rm G}_{\rm SUB}$, 
and all other nodes are in ${\rm Q}$.
Let ${\rm FSSP}_{\rm SUB}(a+b, b)$ be the variation of 
FSSP such that configurations are all 
$A_{a+b, b, l}$ for $l \geq 1$.
(In \cite{Yamashita_et_al_2014}, 
$\alpha$ and $m$ denote $a + b$ and $b$ 
respectively.)

Concerning this variation, Yamashita et al proved 
the following result (\cite{Yamashita_et_al_2014}) :
\begin{quote}
If $a \leq b$ then 
${\rm FSSP}_{\rm SUB}(a+b, b)$ has no minimal-time solutions.
\end{quote}

For each configuration $A_{a+b, b, l}$ of 
${\rm FSSP}_{\rm SUB}(a+b, b)$, let 
$C_{\rm L}(al, bl)$ be the configuration 
in the two-dimensional grid space shown in 
Figure \ref{figure:fig053} (b) 
and let $\text{\rm LSP}[a, b]$ be the variation of FSSP 
such that configurations are all 
$C_{\rm L}(al, bl)$ for $l \geq 1$.

\begin{figure}[htb]
\begin{center}
\includegraphics[scale=1.0]{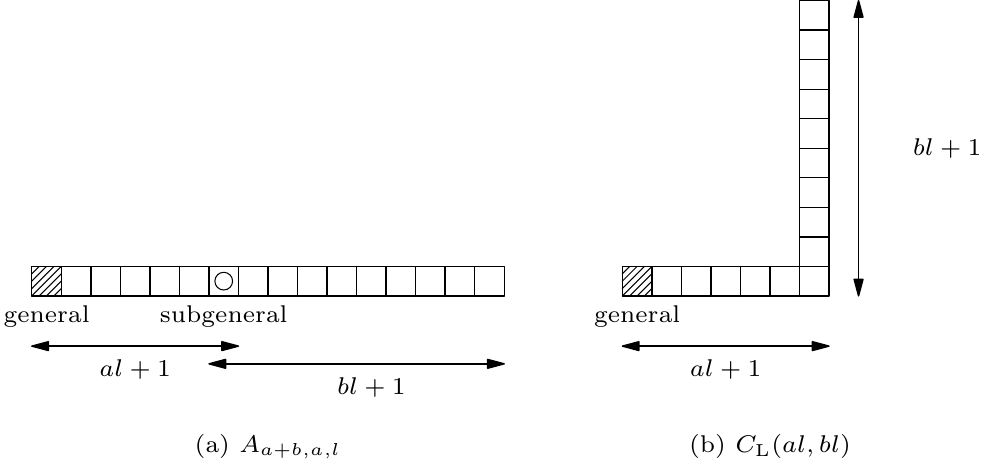}
\end{center}
\caption{(a) a configuration $A_{a+b, b, l}$ of 
${\rm FSSP}_{\rm SUB}(a+b, b)$.
(b) The corresponding configuration 
$C_{\rm L}(al, bl)$ of $\text{\rm LSP}[a, b]$.}
\label{figure:fig053}
\end{figure}

It is easy to see that any solution 
of ${\rm FSSP}_{\rm SUB}(a+b, b)$ can be 
modified to a solution of $\text{\rm LSP}[a, b]$ 
having the same firing time for the 
corresponding $A_{a+b, b, l}$ and 
$C_{\rm L}(al, bl)$, and vice versa. 
Hence the above result can be restated 
for $\text{\rm LSP}[a,b]$ as follows.
\begin{quote}
If $a \leq b$ then 
$\text{\rm LSP}[a, b]$ has no minimal-time 
solutions.
\end{quote}

As far as the author knows, this is the first variation 
(satisfying the four conditions mentioned previously) 
for which we have 
a proof of nonexistence of minimal-time solutions.

\mynewpage

\subsection{The main results}
\label{subsection:main_results}

Now we explain the results of this paper.
Let $C_{\rm L}(w, h)$, 
$C_{\rm RW}(w, h)$, 
$C_{\rm R}(w, h)$ be the configurations 
shown in 
Figure \ref{figure:fig001} (a), (b), (c) 
respectively.
``L'', ``RW'', ``R'' are for 
``L-shaped paths,'' 
``rectangular walls,'' 
and 
``rectangles'' respectively. 
In all of these configurations, 
the general $v_{\rm gen}$ is at the southwest position 
(the box with shadow).

\begin{figure}[htb]
\begin{center}
\includegraphics[scale=1.0]{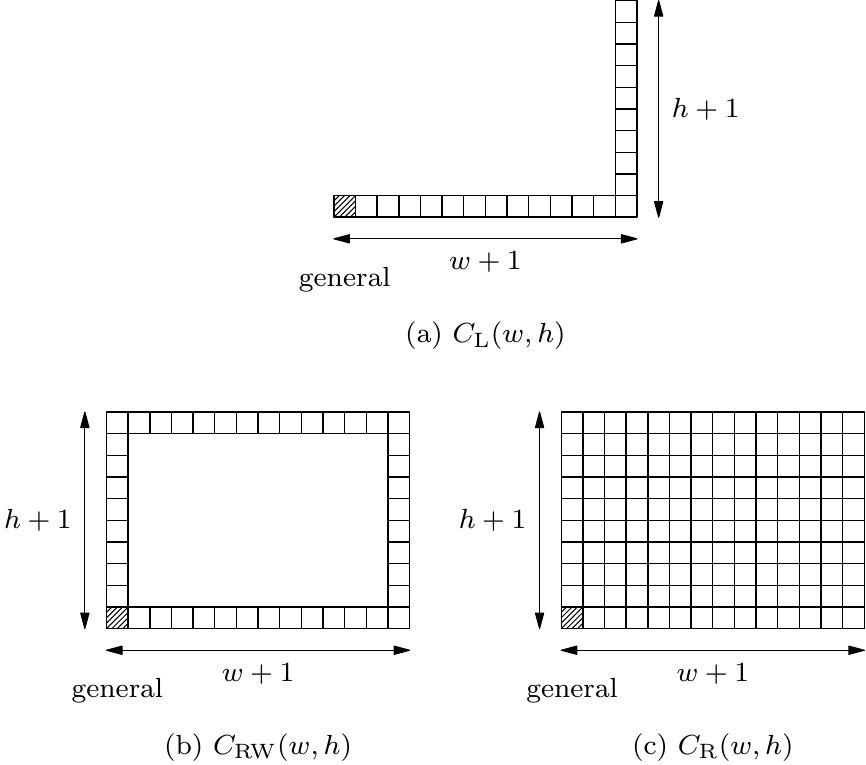}
\end{center}
\caption{Three configurations considered in this paper: 
$C_{\rm L}(w, h)$, $C_{\rm RW}(w, h)$, $C_{\rm R}(w, h)$.}
\label{figure:fig001}
\end{figure}

\noindent
Using these configurations we define $8$ variations 
of FSSP.
They are defined in Table \ref{table:fig027}.

\begin{table}[tbp]
\begin{center}
\includegraphics[scale=1.0]{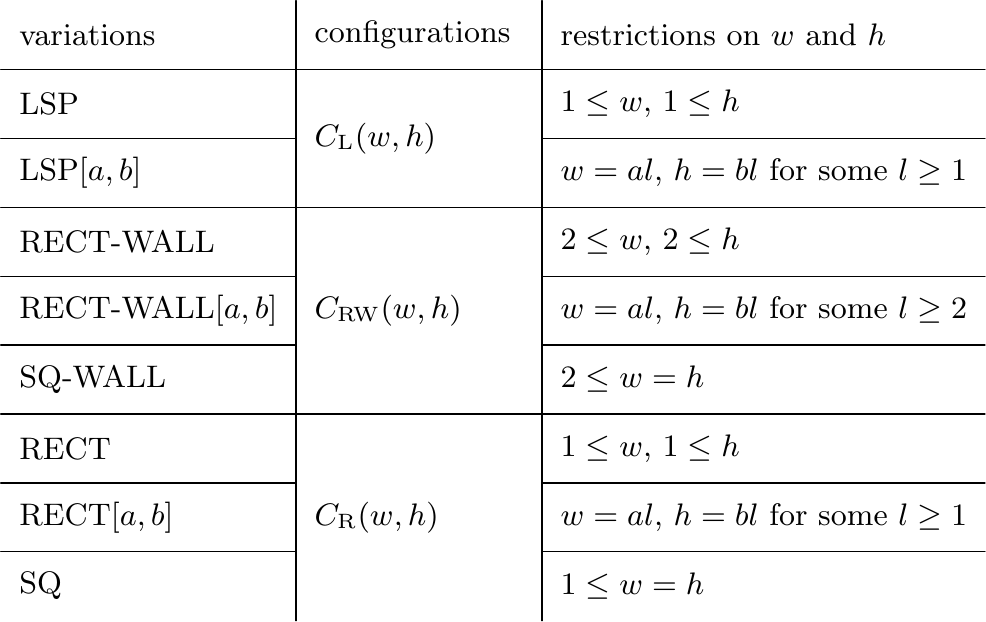}
\end{center}
\caption{Definitions of the $8$ variations of FSSP 
considered in this paper.}
\label{table:fig027}
\end{table}

For example, configurations of $\text{\rm RECT-WALL}[a, b]$ 
are all $C_{\rm RW}(w, h)$ such that 
$w = al$, $h = bl$ for some $l \geq 2$.
``LSP,'' ``RECT-WALL,'' ``SQ-WALL,'' 
``RECT,'' and ``SQ'' are for 
``L-shaped paths,'' 
``rectangular walls,'' 
``square walls,'' 
``rectangles,'' and 
``squares'' respectively.

For each variation listed in Table \ref{table:fig027} 
we also consider its ``generalized'' variation 
in which the general $v_{\rm gen}$ may be at any position 
in configurations.
We denote such a variation by the name obtained by 
adding ``g-'' to the name of the original variation 
(for example ``$\text{\rm g-RECT-WALL}[a,b]$'').

\mynewpage

Our results are summarized as follows.

\medskip

\noindent
(1) For $\Gamma$ $=$ $\text{\rm g-LSP}[a,b]$, 
$\text{\rm RECT-WALL}[a,b]$ 
(including the case $a = b$ \\
($\text{\rm SQ-WALL}$)), 
$\text{\rm g-RECT-WALL}[a,b]$ 
(including the case $a = b$ 
($\text{\rm g-SQ-WALL}$)), 
we determined the value of 
the minimum firing time ${\rm mft}_{\Gamma}(C)$ 
of $\Gamma$ and proved that 
$\Gamma$ has no minimal-time solutions using Yamashita et al's 
idea.

\medskip

\noindent
(2) For $\Gamma$ $=$ 
$\text{\rm RECT-WALL}$, 
$\text{\rm g-RECT-WALL}$, 
$\text{\rm RECT}[a,b]$ 
(excluding the case $a = b$ ($\text{\rm SQ}$)), 
we determined the value of the minimum firing time 
${\rm mft}_{\Gamma}(C)$ of $\Gamma$. 
(However whether $\Gamma$ has minimal-time solutions or not 
remains open.)

\medskip

Our main contribution is the derivation of minimum firing times 
of above mentioned variations.
For results in (1), we need the values of the 
minimum firing times in order to use Yamashita et al's idea.
Even for results in (2) the derived values of 
minimum firing times are useful because they reduce 
the problem of finding minimal-time solutions without any hints 
to the problem of finding solutions having the derived values 
as their firing times.

\mynewpage

In Table \ref{table:tab000}, we summarize known results.
For each variation, in the second column we show 
whether we know its minimum firing time or not and 
in the third column we show whether there exist 
minimal-time solutions for it or not.
The results with ``*'' are our results obtained in this paper 
and marks ``?'' denote open problems.
We will make more comments on these open problems 
in Section \ref{section:conclusion}.

\begin{table}[htb]
\begin{center}
\begin{tabular}{l|l|l}

variations & 
\begin{minipage}{5em}
minimum \\
firing time
\vspace*{0.5ex}
\end{minipage} &
\begin{minipage}{7em}
minimal-time \\
solutions
\vspace*{0.5ex}
\end{minipage}
\\ \hline

LSP & known (\cite{Goto_1962}) 
& exist (\cite{Balzer_1967, Goto_1962, Waksman_1966}) \\

g-LSP & known (\cite{Moore_Langdon_1968}) 
& exist (\cite{Moore_Langdon_1968}) \\

LSP$[a,b]$, $a > b$ & known (\cite{Yamashita_et_al_2014}) 
& exist (\cite{Yamashita_et_al_2014}) \\

LSP$[a,b]$, $a \leq b$ & known (\cite{Yamashita_et_al_2014}) 
& do not exist (\cite{Yamashita_et_al_2014}) \\

g-LSP$[a,b]$ & known (*) 
& do not exist (*) \\ \hline 

RECT-WALL 
& known (*) & ? \\ 

g-RECT-WALL & known (*) & ? \\

RECT-WALL$[a,b]$ ($a \not= b$) & known (*) 
& do not exist (*) \\

g-RECT-WALL$[a,b]$ ($a \not= b$) & known (*) 
& do not exist (*) \\

SQ-WALL & known (*) & do not exist (*) \\

g-SQ-WALL & known (*) & do not exist (*) \\ \hline

RECT & known (\cite{Shinahr_1974}) 
& exist (\cite{Shinahr_1974}) \\

g-RECT & known (\cite{Szwerinski_1982}) 
& exist (\cite{Szwerinski_1982}) \\

RECT$[a,b]$, $a \not= b$ & known (*) & ? \\

g-RECT$[a,b]$, $a \not= b$ & ? & ? \\

SQ & known (\cite{Shinahr_1974}) 
& exist (\cite{Shinahr_1974}) \\

g-SQ & known (\cite{Kobayashi_TCS_2014}) & ? 

\end{tabular}
\end{center}
\caption{Summary of known results on variations 
of FSSP considered in this paper.}
\label{table:tab000}
\end{table}

\mynewpage

Yamashita et al's idea is based on an analysis of 
${\rm st}(v, t, C, A)$ as a function of 
two variables $v$, $t$ for a fixed $C$ 
that satisfies some specific requirements.
We call such an approach an ``automata-theoretical 
approach.''
Their idea can be applied for very special variations.
However, unlike complexity-theoretical approaches, 
it gives nonexistence proofs 
without any assumption from complexity theory such 
as ${\rm P} \not= {\rm NP}$.

If we are to obtain such proofs 
with a complexity-theoretical approach 
then it is necessary to prove that ${\rm mft}(C)$ is 
a computationally difficult function 
without any assumption.
We call this type of proofs ``lower bound proofs'' and 
we know well that lower bound proofs are 
very difficult to obtain 
(and moreover we are gradually understanding why 
it is difficult (\cite{Razborov_1994})).
Therefore, it is unlikely that 
we will have nonexistence proofs 
without any assumption with 
complexity-theoretical approaches in the near future.

In this sense, to study the possibility 
of automata-theoretical approaches 
is very important.

\mynewpage

\section{Notions, notations and technical preliminaries}
\label{section:notions_and_notations}

In this section we explain some basic notions and notations.
We also explain some notions and results that are technical 
but are used throughout the paper.

We assume that in all variations of FSSP we consider in this 
paper, configurations are sets of positions in the 
two-dimensional grid space.
We identify the two-dimensional grid space with 
the set ${\mathbf Z}^{2}$, where 
${\mathbf Z}$ is the set of all integers.
Moreover, we assume that any variation we consider 
in this paper has a solution unless otherwise stated.
(Note that any variation of FSSP that can be modelled as 
a variation of FSSP of connected directed networks 
has a solution (\cite{Kobayashi_JCSS_1978}).
Hence we may say that 
variations having no solutions are pathological.)

Two positions $v=(x, y)$, $v'=(x', y')$ in $\mathbf{Z}^{2}$ 
are {\it adjacent} 
if either $x = x'$ and $|y - y'| = 1$ or 
$|x - x'| = 1$ and $y = y'$.
By a {\it path} we mean a nonempty sequence 
$X = v_{0} v_{1} \ldots v_{n}$ of positions such that 
$v_{i}$, $v_{i+1}$ are adjacent for each $i$ ($0 \leq i \leq n-1$).
We say that $X$ is a path {\it from}\/ $v_{0}$ {\it to}\/ $v_{n}$ 
(or {\it between}\/ $v_{0}$ {\it and}\/ $v_{n}$).
We call the value $n$ (the number of ``steps'' in $X$) 
the {\it length}\/ of $X$ and denote it by $|X|$.

Suppose that $C$ is a connected finite nonempty subset of 
$\mathbf{Z}^{2}$, typically a configuration of a variation 
of FSSP. ($C$ is {\it connected}\/  if there is a path 
from $v$ to $v'$ in $C$ for any $v, v' \in C$.)
For $v, v' \in C$, by the {\it distance}\/ between $v$, $v'$ {\it in} $C$ 
we mean the smallest 
value of $|X|$ for all paths $X$ from $v$ to $v'$ in $C$ 
and denote it by ${\rm d}_{C}(v, v')$.
For $v, v', v''$ in $C$, 
${\rm d}_{C}(v, v'; v'')$ denotes 
${\rm d}_{C}(v, v'') + {\rm d}_{C}(v'', v')$.
When one position in $C$ is specified 
as the position of the general $v_{\rm gen}$, 
by the {\it radius}\/ of $C$ we mean 
$\max \{ {\rm d}_{C}(v_{\rm gen}, v) | v \in C\}$ 
and denote it by ${\rm rad}(C)$.

By $\epsilon_{0}$, $\epsilon_{1}$, $\epsilon_{2}$, $\epsilon_{3}$ 
we denote pairs $(1, 0)$, $(0, 1)$, $(-1, 0)$, $(0, -1)$ 
respectively (corresponding to the directions 
the east, the north, the west, and the south respectively).
By the {\it boundary condition of}\/ $v$ {\it in}\/  $C$ we 
mean the element $(b_{0}, b_{1}, b_{2}, b_{3})$ of $\{0, 1\}^{4}$ 
such that $b_{i}$ is $1$ if the position $v + \epsilon_{i}$ 
is in $C$ and $0$ otherwise.
We denote it by ${\rm bc}_{C}(v)$.

\mynewpage

In Section \ref{section:introduction} we defined 
the notion that 
$\tilde{A}$ is a minimal-time solution of a variation 
$\Gamma$ in two equivalent ways.
The first is by 
\[
(\forall A) (\forall C) [ {\rm ft}(C, \tilde{A}) \leq 
{\rm ft}(C, A) ] 
\]
and the second by 
\[
(\forall C) [ {\rm ft}(C, \tilde{A}) = {\rm mft}_{\Gamma}(C) ], 
\]
where ${\rm mft}_{\Gamma}(C)$ is defined by 
\[
{\rm mft}_{\Gamma}(C) = \min \{ {\rm ft}(C, A) ~|~ 
\text{$A$ is a solution of $\Gamma$} \}.
\]

The first definition has a clear intuitive meaning, that is, 
$\tilde{A}$ is a fastest solution among all solutions.
The second definition is rather technical.
However the function ${\rm mft}_{\Gamma}(C)$ used in 
the definition is essential for our study in this paper.
We use the second definition from now on.

One important point to note is that 
${\rm mft}_{\Gamma}(C)$ is defined whenever  $\Gamma$ has 
a solution irrespective of 
whether $\Gamma$ has minimal-time solutions or not.
Historically, for the original FSSP a minimal-time solution 
of the first definition was found and from this we knew that 
${\rm mft}(C_{n}) = 2n - 2$.
However the proof of ${\rm mft}(C_{n}) = 2n - 2$ itself 
is very easy once we know that there is at least 
one solution.

\mynewpage

We call a finite automaton $A$ a {\it partial solution}
\footnote{The term ``a partial solution of a variation $\Gamma$'' 
is also used for a different meaning 
(for example, in 
\cite{Umeo_Kamikawa_Yunes_partial_2009}).} 
of a variation $\Gamma$ if it satisfies exactly one of the 
following two conditions for each configuration $C$ 
(the selection of the condition may depend on $C$):
\begin{enumerate}
\item[(1)] each $v$ in $C$ never enters the firing state ${\rm F}$,
\item[(2)] there exists a time $t_{C}$ (that may depend on $C$) such 
that ${\rm st}(v, t, C, A) \not= {\rm F}$ for any $v$ in $C$ and 
any time $t$ such that $0 \leq t < t_{C}$ and 
${\rm st}(v, t_{C}, C, A)$ \\
$ = {\rm F}$ for any $v$ in $C$.
\end{enumerate}
By the {\it domain}\/ of a partial solution $A$ we mean the 
set of configurations $C$ for which the case (2) is true, and 
for a configuration $C$ in the domain we denote the time $t_{C}$ 
mentioned in (2) by ${\rm ft}(C, A)$.
We can easily show ${\rm mft}_{\Gamma}(C) \leq {\rm ft}(C, A)$ for any $C$ 
in the domain of $A$.
(Use the finite automaton that simulates both of a solution 
and $A$ and fires when at least one of them fires.)
This result is useful for deriving upper bounds of ${\rm mft}_{\Gamma}(C)$.

\mynewpage

For the definition of FSSP itself we need another 
modified definition which we call 
the {\it boundary sensitive model}\/.
We call the definition of FSSP explained in 
Section \ref{section:introduction} 
the {\it traditional model}\/.

Compared with the traditional model, 
there are two modifications in the boundary sensitive model.
First, instead of one general state ${\rm G}$, 
there may be more than one general state 
${\rm G}_{0}, {\rm G}_{1}, \ldots, {\rm G}_{s-1}$ 
and the state of the general $v_{\rm gen}$ in 
a configuration $C$ at time $0$ is determined by the 
boundary condition $\mathbf{b} = {\rm bc}_{C}(v_{\rm gen})$ of the 
general $v_{\rm gen}$ in $C$.
More precisely, a function $\eta$ from 
$\{0, 1\}^{4}$ to $\{0, 1, \ldots, s-1\}$ is 
specified and 
the state of the general $v_{\rm gen}$ in $C$ at time $0$ 
is ${\rm G}_{\eta(\mathbf{b})}$.
Second, instead of one firing state ${\rm F}$ we specify 
a set ${\mathcal F}$ of firing states.
Any state in ${\mathcal F}$ may be used as a firing state.
The condition which a solution $A$ must satisfy is 
changed to the following: 
for any $n$ ($\geq 1$) there exists a time $t_{n}$ 
such that 
${\rm st}(v, t, C_{n}, A) \not\in {\mathcal F}$ for any node $v$ 
and any time $t$ such that $0 \leq t < t_{n}$ 
and ${\rm st}(v, t_{n}, C_{n}, A) \in {\mathcal F}$ 
for any node $v$.
The state ${\rm Q}$ must not be in $\mathcal{F}$.
However, some general states may be in $\mathcal{F}$ and hence 
the firing time can be $0$.

\mynewpage

The description of FSSP in the traditional model is simple 
and the model is suitable for presenting FSSP 
as an interesting puzzle in automata theory.
However, the boundary sensitive model is 
suitable for design and analysis of solutions.
Concerning this, the essential difference between the two 
models is the following fact: for any nodes $v$, $v'$ in $C$, 
the minimum time for $v'$ to know the boundary condition 
${\rm bc}_{C}(v)$ of $v$ in $C$ is $d = {\rm d}_{C}(v_{\rm gen}, v) + 
{\rm d}_{C}(v, v')$ in the boundary sensitive model 
but $d + 1$ for $v' = v_{\rm gen}$ and $d$ for $v' \not= v_{\rm gen}$ 
in the traditional model.
By the above irregularity in the traditional model, 
design and analysis of solutions are complicated in 
inessential ways in the model.

\mynewpage

We summarize basic results on the relation between the two models 
as a lemma. 
Their proofs are easy and we omit them 
(\cite{Kobayashi_TCS_2014}). 
Let ${\rm mft}_{\Gamma, {\rm bs}}(C)$ and 
${\rm mft}_{\Gamma, {\rm tr}}(C)$ denote the 
minimum firing time ${\rm mft}_{\Gamma}(C)$ 
in the boundary sensitive model and in the traditional model 
respectively.
We call a configuration $C$ a {\it singleton}\/ if 
$C = \{ v_{\rm gen} \}$.
A configuration $C$ is a singleton 
if and only if ${\rm bc}_{C}(v_{\rm gen}) = (0, 0, 0, 0)$.
(This is true by our assumption that $C$ is 
a subset of ${\mathbf Z}^{2}$.
This is not true, for example, for the FSSP of 
directed or undirected networks.)

\begin{lem}
\label{lemma:relations_between_bs_and_tr}

\noindent
\begin{enumerate}
\item[{\rm (1)}] For any configuration $C$, 
${\rm mft}_{\Gamma, {\rm bs}}(C) \leq 
{\rm mft}_{\Gamma, {\rm tr}}(C) \leq 
{\rm mft}_{\Gamma, {\rm bs}}(C) + 1$.
\item[{\rm (2)}] For a singleton  $C$, 
${\rm mft}_{\Gamma, {\rm bs}}(C) = 0$ and 
${\rm mft}_{\Gamma, {\rm tr}}(C) = 1$.
\item[{\rm (3)}] If ${\rm rad}(C) + 1 \leq {\rm mft}_{\Gamma, {\rm bs}}(C)$ 
for any nonsingleton $C$, then 
${\rm mft}_{\Gamma, {\rm tr}}(C) =
{\rm mft}_{\Gamma, {\rm bs}}(C)$ for any nonsingleton $C$.
In this case, $\Gamma$ has a minimal-time solution in the 
boundary sensitive model if and only if it has one in the 
traditional model.
\item[{\rm (4)}] If ${\rm bc}_{C}(v_{\rm gen})$ is the same 
for all nonsingleton $C$, then 
${\rm mft}_{\Gamma, {\rm tr}}(C)$ $=$ 
${\rm mft}_{\Gamma, {\rm bs}}(C)$ for all nonsingleton $C$.
In this case, $\Gamma$ has a minimal-time solution in the 
boundary sensitive model if and only if it has one in the 
traditional model.
\end{enumerate}
\end{lem}

\mynewpage

By (3), (4) of the lemma and the fact that 
any configuration of the variations considered in this paper 
is nonsingleton, we know that the difference of the two models 
appears only for the case where a variation is 
a generalized variation (having ``g-'' in its name) 
and ${\rm mft}(C) = {\rm rad}(C)$ 
for some of its configurations $C$.
We have this situation for the three variations 
$\text{\rm g-LSP}[a,b]$, $\text{\rm g-RECT-WALL}[a,b]$, 
$\text{\rm g-RECT}[a,b]$.
For these variations we cannot escape from dealing with 
the inessential irregularity in the traditional model 
mentioned above.
(See also \cite{Kobayashi_Goldstein_UC_2005} for the comparison 
of the two models.)

\mynewpage

Let $\Gamma$ be a variation.
For each configuration $C$ of $\Gamma$, 
each node $v = (i, j)$ in $C$ and 
each time $t$, 
we define 
the {\it available information}\/ of $v$ in $C$ at time $t$ 
as follows.
If ${\rm d}_{C}(v_{\rm gen}, v) > t$ then the available information 
is the letter ``Q''.
Otherwise the available information is the pair 
$(t, ((i - r, j - s), X))$.
Here $(r, s)$ is the position of $v_{\rm gen}$ 
and $X$ is the set of all elements 
$((i' - r, j' - s), {\rm bc}_{C}((i', j')))$ 
such that 
\begin{enumerate}
\item[(1)] $(i', j') \in C$.
\item[(2)] ${\rm d}_{C}(v_{\rm gen}, v; (i', j')) \leq t$.
\end{enumerate}
By ${\rm ai}(v, t, C)$ we denote the available information 
of $v$ in $C$ at $t$.

Intuitively, if ${\rm d}_{C}(v_{\rm gen}, v) \leq t$ 
then ${\rm ai}(v, t, C)$ is the set of all information 
the node $v$ can get at time $t$ concerning 
the current time and 
the structure of the configuration $C$ in which it is. 
In the design of distributed algorithms, 
the second component $((i - r, j - s), X)$ of 
${\rm ai}(v, t, C)$ is usually called 
the {\it local map}\/ of $v$ at $t$ concerning $C$.
Hence we call it the {\it local map component}\/ 
of ${\rm ai}(v, t, C)$.
We call the first component $t$ of ${\rm ai}(v, t, C)$ 
its {\it time component}.

\mynewpage

In Figure \ref{figure:fig048} (a) we show a configuration $C$. 
The origin $(0, 0)$ is at the southwest corner and 
$v_{\rm gen}$ is at $(r, s) = (2, 3)$. 
As an example let us consider 
${\rm ai}((4, 0), 9, C)$.
$v = (4, 0)$ is the node marked with a bullet in 
the figure (a).
${\rm ai}((4, 0), 9, C)$ is of the form 
$(9, ((2, -3), X))$.
$X$ contains $14$ elements of the form 
$((i' - r, j' - s), {\rm bc}_{C}((i', j')))$ 
for $(i', j')$ in 
$\{ (0, 0)$, $(1, 0)$, $(2, 0)$, \ldots, 
$(2, 5)$, $(3, 0)$, $(4, 0)$, \ldots, 
$(4, 2)$, $(5, 0)$, $(6, 0) \}$.
Hence 
$X = \{ ((-2, -3), (1, 1, 0, 0))$, 
$((-1, -3), (1, 0, 1, 0))$, \ldots, 
$((4, -3), (0, 1, 1, 0)) \}$. 

\begin{figure}[htb]
\begin{center}
\includegraphics[scale=1.0]{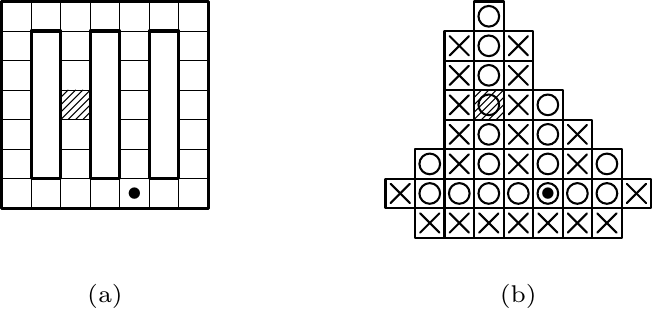}
\end{center}
\caption{(a) A configuration $C$. 
(b) The figure that represents  
the local map component of ${\rm ai}((4, 0), 9, C)$.}
\label{figure:fig048}
\end{figure}

The definition of available information 
${\rm ai}(v, t, C)$ is conceptually simple.
However its local map component $((i - r, j - s), X)$ 
has many redundancies in it. 
We can represent it with the following figure 
more succinctly. 
We write a square with shadow at $(0, 0)$.
We write a square with a bullet at $(i - r, j - s)$. 
For each element $((i' - r, j' - s), (b_{0}, b_{1}, b_{2}, b_{3}))$ 
of $X$ we write a square with a circle at $(i' - r, j' - s)$. 
Moreover, for each $d$ ($0 \leq d \leq 3$) we 
write a square at $(i' - r, j' - s) + \epsilon_{d}$, 
with a circle if $b_{d} = 1$ and with a cross if 
$b_{d} = 0$.
(Remember that 
$\epsilon_{0}, \epsilon_{1}, \epsilon_{2}, \epsilon_{3}$ are 
$(1, 0), (0, 1), (-1, 0), (0, -1)$.)

The figure has a box at a position $w$ if and only if 
either $w + v_{\rm gen} \in C$ and 
${\rm d}_{C}(v_{\rm gen}, v; w + v_{\rm gen}) 
\leq t$ or there is a position $w'$ such that 
$w' + v_{\rm gen} \in C$, 
${\rm d}_{C}(v_{\rm gen}, v; w' + v_{\rm gen}) \leq t$ and 
$w + v_{\rm gen}$, $w' + v_{\rm gen}$ are adjacent.

In Figure \ref{figure:fig048} (b) we show this figure 
for ${\rm ai}((4, 0), 9, C)$ for the configuration $C$ 
in Figure \ref{figure:fig048} (a).
From the time component $t = 9$ and 
this figure representing the local map component 
we can reconstruct 
the original ${\rm ai}(v, t, C)$ completely.

\mynewpage

In \cite{Kobayashi_TCS_2014} several results on 
${\rm ai}(v, t, C)$ were shown and proved as 
Fact 1 -- Fact 6.
(In \cite{Kobayashi_TCS_2014}, 
${\rm ai}(v, t, C)$ was 
the triple $(t, (i - r, j - s), X)$ instead of 
the pair $(t, ((i - r, j - s), X))$.)
Here we use Fact 3.
\medskip

\noindent
{\bf Fact 3:} (For both of the boundary sensitive 
and the traditional models.) If ${\rm ai}(v, t, C) = {\rm ai}(v', t, C')$ 
then ${\rm st}(v, t, C, A) = {\rm st}(v', t, C', A)$ 
for any solution $A$.

\medskip

\noindent
The following lemma is one simple case of Fact 4.

\begin{lem} {\rm (}For both of the boundary sensitive 
and the traditional models.{\rm )}
Let $C$, $C'$ be configurations of $\Gamma$, $v \in C$, 
$v' \in C'$ be nodes, and $t$ be a time.
If ${\rm ai}(v, t, C) = {\rm ai}(v', t, C')$ and 
${\rm rad}(C') > t$, then 
${\rm mft}_{\Gamma}(C) > t$.
\label{lemma:fact_4}
\end{lem}

\begin{proof}
We show the proof for the boundary sensitive model.
Let $A$ be an arbitrary solution of $\Gamma$ and let $v''$ be 
a node in $C'$ such that ${\rm d}_{C'}(v_{\rm gen}, v'') > t$.
Then ${\rm st}(v'', t, C', A) = {\rm Q}$ 
and hence 
${\rm st}(v', t, C', A) \not\in \mathcal{F}$ 
because $A$ is a solution.
By Fact 3, ${\rm st}(v, t, C, A) \not\in \mathcal{F}$.
Therefore, $A$ cannot fire $C$ at $t$.
$A$ was an arbitrary solution.
Hence ${\rm mft}_{\Gamma}(C) > t$.
\end{proof}

\mynewpage

\section{Variations of FSSP for L-shaped paths}
\label{section:l_shaped_paths}

In this section we consider the four variations of FSSP for 
L-shaped paths, that is, 
LSP, LSP$[a,b]$ (see Tables \ref{table:fig027}) and their 
generalized variations g-LSP, g-LSP$[a,b]$.
All configurations of these variations are L-shaped paths 
$C_{\rm L}(w,h)$ 
(see Figure \ref{figure:fig001} (a)).
We assume that 
$p_{0}, p_{1}, \ldots, p_{w-1}$, 
$p_{w}, \ldots, p_{w+h-1}, p_{w+h}$ 
are nodes in $C_{\rm L}(w,h)$, 
$p_{0}$ being the west terminal, 
$p_{w}$ being the corner, and 
$p_{w+h}$ being the north terminal.
When we consider $C_{\rm L}(w, h)$ as a configuration 
of generalized variations $\text{\rm g-LSP}$, 
$\text{\rm g-LSP}[a, b]$ 
such that $p_{i}$ is the 
position of the general, we denote it by 
$C_{\rm L}(w, h, i)$.
We may simply write $C(w, h)$, $C(w, h, i)$ 
instead of $C_{\rm L}(w, h)$, $C_{\rm L}(w, h, i)$.

\mynewpage

\subsection{The variations {\rm LSP} and {\rm g-LSP}}
\label{subsection:l_shaped_no_restrictions}

The following two theorems (Theorems 
\ref{theorem:lsp}, 
\ref{theorem:g_lsp}) are obvious.
However we show and prove them as examples of 
how to use Lemma \ref{lemma:fact_4} 
to obtain lower bounds of ${\rm mft}(C)$.

\mynewpage

\begin{thm}
\label{theorem:lsp}

\noindent
\begin{enumerate}
\item[{\rm (1)}] ${\rm mft}_{\rm LSP}(C_{\rm L}(w,h)) = 
2(w + h)$.
\item[{\rm (2)}] ${\rm LSP}$ has a minimal-time solution.
\end{enumerate}
\end{thm}

\begin{proof} 
First we prove the lower bound 
$2(w + h) \leq {\rm mft}(C(w,h))$ using 
Lemma \ref{lemma:fact_4}.
Let $C$ denote $C(w, h)$.

Intuitively this lower bound is explained as follows.
At time $2(w + h) - 1$ the node $p_{0}$ in $C$ can know 
the position of the corner $p_{w}$ 
but cannot know the position of the north terminal 
$p_{w+h}$.
Hence $p_{0}$ cannot exclude the possibility that 
it is in a configuration $C' = C(w, h')$ such that 
$h < h'$ and ${\rm rad}(C(w, h')) > 2(w + h) - 1$.
In that case $p_{0}$ should not fire because 
some node in $C'$ (for example, the north terminal $p_{w+h'}$) 
is in ${\rm Q}$.
Based on this inference $p_{0}$ cannot fire at time 
$2(w + h) - 1$.

This intuitive explanation can be rewritten as a formal 
proof using Lemma \ref{lemma:fact_4}.
Consider the case $C = C(3, 1)$ (ant hence 
$2(w + h) = 8$).
Select $C(3, 5)$ as $C'$.
Then we have ${\rm ai}(p_{0}, 7, C) = {\rm ai}(p_{0}, 7, C')$ 
and ${\rm rad}(C') = 8 > 7$.
Therefore we have ${\rm mft}(C) > 7$ by the lemma.

In Figure \ref{figure:fig049} (a), (b), (c) 
we show $C = C(3, 1)$, $C' = C(3, 5)$, and 
the local map component of the common available information 
respectively.
In $C'$ in the figure (b) we mark a node $v$ 
such that ${\rm d}(v_{\rm gen}, v) > 7$ with ``*''.

\begin{figure}[htb]
\begin{center}
\includegraphics[scale=1.0]{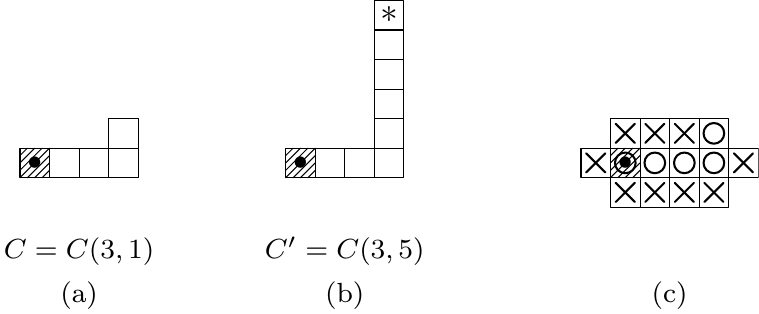}
\end{center}
\caption{(a) shows the configuration $C$ for which 
we prove the lower bound and (b) shows 
another configuration $C'$ that is used 
in applying Lemma \ref{lemma:fact_4}.
(c) shows the local map component of 
${\rm ai}(p_{0}, 7, C) = {\rm ai}(p_{0}, 7, C')$.
}
\label{figure:fig049}
\end{figure}

As for the upper bound, 
we can modify a minimal-time solution of the original FSSP 
(the FSSP of straight lines with $v_{\rm gen}$ at the left ends) 
to a solution of LSP without changing the 
firing time.
Therefore we have a solution that fires $C(w,h)$ at time 
$2(w + h)$ and hence ${\rm mft}(C(w, h)) \leq 2(w + h)$.
This and the above lower bound prove (1) and (2) of the theorem.
\end{proof}

\mynewpage

Next we consider g-LSP.
For this variation we must distinguish the two models of FSSP 
because the condition (4) of Lemma 
\ref{lemma:relations_between_bs_and_tr} is not satisfied.

\begin{thm}
\label{theorem:g_lsp}

\noindent
\begin{enumerate}
\item[{\rm (1)}] 
${\rm mft}_\text{\rm g-LSP}(C_{\rm L}(w, h, i)) 
= \max \{ w + h + i, 2w + 2h - i \}$.
\item[{\rm (2)}] {\rm g-LSP} has a minimal-time solution.
\end{enumerate}
\end{thm}

\begin{proof}
Intuitively the two lower bounds 
$w + h + i \leq {\rm mft}_{\rm bs}(C(w, h, i))$, 
$2w + 2h - i \leq {\rm mft}_{\rm bs}(C(w, h, i))$ 
for the boundary sensitive model 
are obvious because 
at time $(w + h + i) - 1$ the north terminal $p_{w+h}$ 
cannot know the position of the west terminal $p_{0}$ 
and 
at time $(2w + 2h - i) - 1$ the west terminal $p_{0}$ 
cannot know the position of the north terminal $p_{w+h}$.
We can translate these intuitive explanations into formal 
proofs using Lemma \ref{lemma:fact_4}.

As for the upper bound 
${\rm mft}_{\rm tr}(C(w, h, i)) \leq 
\max \{ w + h + i, 2w + 2h - i \}$ for 
the traditional model, 
we can construct a solution of g-LSP 
of the traditional model 
that fires 
$C(w, h, i)$ at 
$\max \{ w + h + i, 2w + 2h - i \}$ 
by modifying the minimal-time solution of 
the original g-FSSP by Moore and Langdon 
(\cite{Moore_Langdon_1968}).
Hence we have this upper bound.

Combining these results and the relation 
${\rm mft}_{\rm bs}(C(w, h, i)) \leq {\rm mft}_{\rm tr}(C(w,$ \\
$ h, i))$
we have (1), (2) of the theorem.
\end{proof}

\mynewpage

\subsection{The variation ${\rm LSP}[a,b]$}
\label{subsection:l_a_b}

In this subsection we consider the variation 
${\rm LSP}[a,b]$ and show three results 
by Yamashita et al (\cite{Yamashita_et_al_2014}) 
reformulated as results on this variation.
For two of them (Theorems \ref{theorem:mft_lsp_a_b}, 
\ref{theorem:yamashita_result_2})
we also show proofs.

First we determine the value 
${\rm mft}_{{\rm LSP}[a,b]}(C_{\rm L}(w,h))$.

\begin{thm}[Theorem 1 in \cite{Yamashita_et_al_2014}]
\label{theorem:mft_lsp_a_b}
${\rm mft}_{\text{\rm LSP}[a,b]}(C_{\rm L}(w,h))$ is 
$2w$ for  $a > b$ and $w + h$ for $a \leq b$.
\end{thm}

\begin{proof}

(1) Proof for $a > b$.
Intuitively the lower bound 
$2w \leq {\rm mft}(C(w, h))$ is obvious 
because at time $2w - 1$ the node $p_{0}$ 
cannot know the position of the corner $p_{w}$.
We can translate this informal explanation into 
a formal proof using Lemma \ref{lemma:fact_4}.

\mynewpage

As for the proof of the upper bound, 
we consider the case $a = 3$, $b = 2$, 
$w = 6$, $h = 4$ as an example and 
prove ${\rm mft}_{\text{\rm LSP}[3,2]}(C(6, 4)) \leq 12$.
For this we construct 
a partial solution $A$ with the domain 
$\{C(6,4)\}$ that fires $C(6,4)$ at time $12$.
First we explain $A$ intuitively.

Suppose that copies of $A$ are placed on a configuration 
$C = C(w, h)$ of $\text{\rm LSP}[3, 2]$.
$A$ uses a signal ${\rm R}$ to check the statement $w = 6$.
If the statement is true then the signal broadcasts a message 
``$w = 6$'' to all nodes and vanishes.
If the statement is not true then the signal simply vanishes.
A more detailed explanation of the behavior of the signal ${\rm R}$ 
is as follows.

The signal ${\rm R}$ starts at the general node $p_{0}$ at time $0$ 
and proceeds to the east $6$ steps to arrive at $p_{6}$.
If it fails to arrive at $p_{6}$ then $w < 6$ and the signal vanishes.
If it arrives at $p_{6}$ but the boundary condition of $p_{6}$ is 
not $(0, 1, 1, 0)$ (the boundary condition of the southeast corner) then 
$6 < w$ and the signal vanishes.
Otherwise the statement $w = 6$ is true 
and hence the signal ${\rm R}$ 
broadcasts the message ``$w = 6$'' to all nodes and vanishes.

We design $A$ so that a node fires at time $12$ if 
it has received the message ``$w = 6$'' before or at time $12$.
Suppose that $C(w, h) = C(6, 4)$.
Then the message is generated at $p_{6}$ at time $6$ and hence 
all nodes receive the message before or at time $12$ ($= 6 + 6$) 
because ${\rm d}_{C(6, 4)}(p_{6}, v) \leq 6$ for any node $v$ 
in $C(6, 4)$. Hence all nodes fire at time $12$.
Conversely suppose that a node in $C(w, h)$ fires at some time.
This means that the message ``$w = 6$'' was generated, 
the statement $w = 6$ is true, and hence $C(w, h) = C(6, 4)$.
Therefore, $A$ is a partial solution that has $\{C(6, 4)\}$ 
as its domain and that fires $C(6, 4)$ at time $12$.
This shows the desired upper bound ${\rm mft}(C(6, 4)) \leq 12$.

A formal definition of $A$ is as follows.
$A$ uses six states ${\rm R}_{0}$, ${\rm R}_{1}$, \ldots, 
${\rm R}_{5}$ to simulate the travel of the signal ${\rm R}$ and 
seven states ${\rm S}_{6}$, ${\rm S}_{7}$, \ldots, 
${\rm S}_{12}$ to simulate the generation and the 
propagation of the message.
${\rm R}_{0}$ plays the role of the general state ${\rm G}$ 
and ${\rm S}_{12}$ plays the role of the firing state ${\rm F}$.
$A$ has also the state ${\rm Q}$ and hence $A$ has 
$6 + 7 + 1 = 14$ states.

Let $\delta(s, s_{0}, \ldots, s_{3})$ denote the state of 
a node $v$ at time $t + 1$ when the state of $v$ is $s$ 
and the values of the inputs to $v$ from the east, the north, the west, 
and the south are $s_{0}, \ldots, s_{3}$ respectively at time $t$.
The function $\delta$ is defined as follows:
\begin{enumerate}
\item[(1)] $\delta({\rm Q}, {\rm Q}, \#, {\rm R}_{i}, \#) 
= {\rm R}_{i+1}$ ($0 \leq i \leq 4$),
\item[(2)] $\delta({\rm Q}, \#, {\rm Q}, {\rm R}_{5}, \#) = {\rm S}_{6}$,
\item[(3)] $\delta(s, s_{0}, \ldots, s_{3}) = {\rm S}_{i+1}$, 
($6 \leq i \leq 11$, 
all of $s, s_{0}, \ldots, s_{3}$ are in 
$\{{\rm S}_{i}, {\rm Q}, \#\}$, 
and at least one of them is ${\rm S}_{i}$),
\item[(4)] $\delta(s, s_{0}, \ldots, s_{3}) = {\rm Q}$ for all other 
combinations of $s, s_{0}, \ldots, s_{3}$.
\end{enumerate}
The rule (1) simulates the travel of the signal ${\rm R}$, 
the rule (2) simulates the generation of the message, 
and the rule (3) simulates the propagation of the message.
The index $i$ of ${\rm R}_{i}$, ${\rm S}_{i}$ denotes the current time.

\mynewpage

\noindent
(2) Proof for $a \leq b$. 
The lower bound $w + h \leq {\rm mft}(C(w,h))$ 
is obvious because 
$w + h = {\rm rad}(C(w,h)) \leq {\rm mft}(C(w,h))$.
For the upper bound 
${\rm mft}(C(w, h)) \leq w + h$ 
we use the case 
$a = 3$, $b = 4$, $w = 6$, $h = 8$ 
as an example and prove ${\rm mft}(C(6, 8)) \leq 14$.
We modify the construction of the partial solution $A$ in (1) 
as follows.

The signal ${\rm R}$ checks the statement $w = 6$ and 
broadcasts a message ``$w = 6$'' to all nodes if the statement is true.
A node fires at time $14$ if it has received the message before or at 
time $14$.
(Formally, $A$ uses $16$ states 
${\rm R}_{0}$, ${\rm R}_{1}$, \ldots, ${\rm R}_{5}$, 
${\rm S}_{6}$, \ldots, ${\rm S}_{14}$, ${\rm Q}$.)
In this case, if the given configuration $C(w, h)$ is 
$C(6, 8)$ then the message ``$w = 6$'' 
is generated at $p_{6}$ at time $6$ 
and any node $v$ in $C(6, 8)$ receives the message before or at 
time $14$ ($= 6 + 8$) because ${\rm d}_{C(6, 8)}(p_{6}, v) \leq 8$ 
for any node $v$ in $C(6, 8)$.
Thus we obtain a partial solution that has the domain $\{C(6, 8)\}$ 
and that fires $C(6, 8)$ at time $14$.
\end{proof}

The above idea to prove upper bounds of ${\rm mft}(C)$ 
is used repeatedly in this paper.
In the proof we construct a partial solution $A$ that has 
$C$ in its domain.
In $A$, several signals start at $v_{\rm gen}$ at time $0$,  
proceed with speed $1$ and check some statements concerning 
the structure of the configuration.
If a signal checks a statement and the statement is true, 
the signal broadcasts a message to all nodes. 
Each node fires at a predetermined time if 
a certain condition on messages that have been received by the node 
before or at the time is satisfied.
From now on, we call a partial solution constructed using this idea 
a ``{\it check and broadcast partial solution}.''

\mynewpage

Of the following two theorems, the latter is 
relevant to this paper.

\begin{thm}[Algorithm 2 in \cite{Yamashita_et_al_2014}]
\label{theorem:yamashita_result_1}
If $a > b$ then ${\rm LSP}[a,b]$ has 
a minimal-time solution.
\end{thm}

\begin{thm}[Theorem 2 in \cite{Yamashita_et_al_2014}]
\label{theorem:yamashita_result_2}
If $a \leq b$ then ${\rm LSP}[a,b]$ has no 
minimal-time solutions.
\end{thm}

\begin{proof}
We assume that $\tilde{A}$ is a solution of $\text{\rm LSP}[a,b]$ 
that fires $C(w,h)$ 
at time $w + h$ ($= {\rm mft}(C(w,h))$) 
for any configuration $C(w,h)$ of ${\rm LSP}[a,b]$, 
and derive a contradiction.
For any $i$ ($0 \leq i \leq w+h$) and 
any $t$, by $s_{i}^{t}$ we denote the state 
${\rm st}(p_{i}, t, C(w,h), \tilde{A})$.
($s_{i}^{t}$ depends also on implicitly understood 
$w$, $h$.)
Note that for any $t$ ($\leq w + h$), only nodes $p_{i}$ 
with $0 \leq i \leq t$ may be in nonquiescent states.

$\tilde{A}$ has only a finite number of states.
Therefore if $w$, $h$ are sufficiently large 
there are two times $t_{0}$, $t_{1}$ such that 
$w+1 \leq t_{0} < t_{1} \leq w+h-1$ and 
$s_{t_{0}-1}^{t_{0}} s_{t_{0}}^{t_{0}} 
= s_{t_{1}-1}^{t_{1}} s_{t_{1}}^{t_{1}}$.
Let $w,h,t_{0},t_{1}$ be such values.

If $t_{1} \leq w + h - 2$, we have 
$s_{t_{0}+1}^{t_{0}} = s_{t_{0}+2}^{t_{0}} = {\rm Q}$ and 
$s_{t_{1}+1}^{t_{1}} = s_{t_{1}+2}^{t_{1}} = {\rm Q}$, 
and hence we have 
$s_{t_{0}}^{t_{0}+1} s_{t_{0}+1}^{t_{0}+1} = 
s_{t_{1}}^{t_{1}+1} s_{t_{1}+1}^{t_{1}+1}$.
In Figure \ref{figure:fig005} (a) we explain this inference with a figure 
that depicts two parts of a configuration in two consecutive times 
side by side.
Two diamonds in two corresponding positions 
$p_{i}$ and $p_{i+t_{1}-t_{0}}$ in the figure ($p_{t_{0}}$ and $p_{t_{1}}$, for example) 
represent two states that must be identical.
A letter ``${\rm Q}$'' at a node means that the node 
really exists in the configuration and that 
its state is ${\rm Q}$.
At these positions in the configuration, 
paths $\ldots p_{t_{0}-1} p_{t_{0}} p_{t_{0}+1} \ldots $ 
and $\ldots p_{t_{1}-1} p_{t_{1}} p_{t_{1}+1} \ldots $ 
proceed vertically (from the south to the north).
However, in Figure \ref{figure:fig005} we draw them horizontally 
to save the space.

\begin{figure}[htb]
\begin{center}
\includegraphics[scale=1.0]{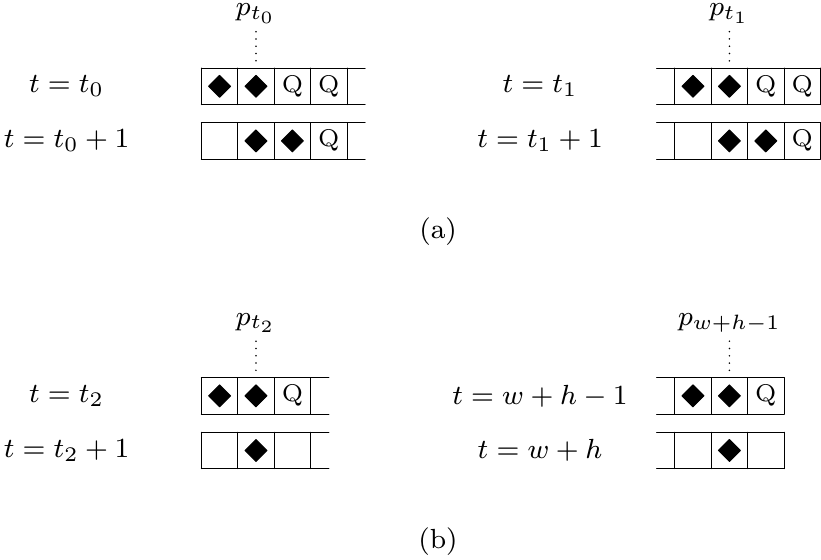}
\end{center}
\caption{Intuitive explanations of the two inferences in 
Theorem \ref{theorem:yamashita_result_2}.}
\label{figure:fig005}
\end{figure}

Repeating this inference, 
we can show 
$s_{t_{2}-1}^{t_{2}} s_{t_{2}}^{t_{2}} = 
s_{w+h-2}^{w+h-1} s_{w+h-1}^{w+h-1}$, 
where $t_{2}$ is the time 
$t_{2} = w + h - 1 - t_{1} + t_{0}$.
Moreover we have 
$s_{t_{2}+1}^{t_{2}} = {\rm Q}$, 
$s_{w+h}^{w+h-1} = {\rm Q}$.
Hence we have 
$s_{t_{2}}^{t_{2}+1} = 
s_{w+h-1}^{w+h}$. 
Again, in Figure \ref{figure:fig005} (b) we 
explain this inference with a figure.
But $s_{w+h-1}^{w+h} = {\rm F}$ 
because $\tilde{A}$ fires $C(w,h)$ at 
time $w+h$.
This means that 
$s_{t_{2}}^{t_{2}+1} = {\rm F}$ 
and hence 
the node $p_{t_{2}}$ is in ${\rm F}$ at 
the time $t_{2} + 1 = w + h - (t_{1} - t_{0}) < 
w + h$.
This is a contradiction.
\end{proof}

\mynewpage

Theorem \ref{theorem:yamashita_result_2} gives 
an interesting result.
For a configuration $C$ of a variation $\Gamma$ we define 
the {\it smallest solution size of $C$ in $\Gamma$}\/ 
(denoted by ${\rm sss}_{\Gamma}(C)$) by
\begin{eqnarray*}
{\rm sss}_{\Gamma}(C) & = & 
\min \{ \text{the number of states of $A$} ~|~ \\
& & \quad\quad \text{$A$ is a solution of $\Gamma$ that fires $C$ 
at ${\rm mft}_{\Gamma}(C)$} \}.
\end{eqnarray*}
This value is well-defined because we assume that $\Gamma$ has 
a solution.

\begin{thm}
\label{theorem:existence_and_mns}
If a variation $\Gamma$ of FSSP has a solution the following 
are equivalent.
\begin{enumerate}
\item[{\rm (1)}] $\Gamma$ has a minimal-time solution.
\item[{\rm (2)}] There is a constant $c$ such that 
${\rm sss}_{\Gamma}(C) \leq c$ for any configuration $C$ of 
$\Gamma$.
\end{enumerate}
\end{thm}
\noindent
The proof of this theorem is easy and we omit it.
By Theorems \ref{theorem:yamashita_result_2}, 
\ref{theorem:existence_and_mns} 
we know that 
if $a \leq b$ then 
${\rm sss}_{\text{\rm LSP}[a,b]}(C)$ is unbounded.
This is the first example of variations for which 
we know this.

\begin{thm}
\label{theorem:mns_lsp_a_b}
If $a \leq b$ then 
\[
\sqrt{h - 1} \leq {\rm sss}_{\text{\rm LSP}[a,b]}(C_{\rm L}(w, h)) 
\leq 6(w + h + 2).
\]
\end{thm}

\begin{proof}
Lower bound. In the proof of 
Theorem \ref{theorem:yamashita_result_2}, we considered 
the sequence $s_{t-1}^{t} s_{t}^{t}$ for 
$t = w + 1, w + 2, \ldots, w + h - 1$ and the number of 
possible values of $t$ is $h - 1$.
Therefore, if a minimal-time solution $A$ has $q$ states 
and $q^{2} < h - 1$ then we can derive a contradiction 
by the proof of that theorem.
Therefore we have $h - 1 \leq q^{2}$ for any minimal-time 
solution $A$.

\medskip

\noindent
Upper bound. In the proof of Theorem \ref{theorem:mft_lsp_a_b} 
we constructed a partial solution $A$ with $w_{0} + h_{0} + 2$ 
states that fires $C(w_{0}, h_{0})$ at time 
$w_{0} + h_{0} = {\rm mft}(C(w_{0}, h_{0}))$.
Moreover, we can modify the six-state solution $A'$ of 
the original FSSP by Mazoyer (\cite{Mazoyer_1987}) 
to a six-state solution of $\text{\rm LSP}[a,b]$.
By simulating both of $A$, $A'$ simultaneously we 
can construct a solution of $\text{\rm LSP}[a,b]$ that has 
$6(w_{0} + h_{0} + 2)$ states and fires $C(w_{0}, h_{0})$ 
at time ${\rm mft}(C(w_{0}, h_{0}))$.
\end{proof}

\mynewpage

\subsection{The variation $\text{\rm g-LSP}[a,b]$}
\label{subsection:g_l_a_b}

In this subsection we consider 
the variation $\text{g-LSP}[a,b]$.
In this variation a configuration is $C_{\rm L}(w, h, i)$ 
such that $w = al$, $h = bl$ for some $l \geq 1$ and 
$0 \leq i \leq w + h$.

First we determine the value 
${\rm mft}_{\text{\rm g-LSP}[a,b]}(C)$.
Then using this result and 
Theorem \ref{theorem:yamashita_result_2} we show that 
$\text{\rm g-LSP}[a,b]$ has no minimal-time solutions 
for any $a, b$ (not necessarily $a \leq b$).

\mynewpage

\begin{lem}
\label{lemma:from_bs_to_tr}
Suppose that $A$ is a partial solution of 
a variation $\Gamma$ of $\text{\rm FSSP}$ 
of the boundary sensitive model 
and that 
${\rm rad}(C) + 1 \leq {\rm ft}(C, A)$ 
for any configuration $C$ in the domain of $A$.
Then there is a partial solution $A'$ of $\Gamma$ of the 
traditional model that has the same domain and the same firing times 
as $A$.
\end{lem}

\begin{proof}

For each element $\mathbf{b} \in \{0, 1\}^{4}$, 
let $A_\mathbf{b}$ be the finite automaton of 
the traditional model that is obtained from $A$ 
by specifying the state ${\rm G}_{\eta(\mathbf{b})}$ 
as its general state.
(Recall that $A$ has general states ${\rm G}_{0}, 
{\rm G}_{1}, \ldots, {\rm G}_{s-1}$ and 
when the boundary condition 
${\rm bc}_{C}(v_{\rm gen})$ of $v_{\rm gen}$ 
in a configuration $C$ is $\mathbf{b}$ then 
${\rm G}_{\eta(\mathbf{b})}$ is used as the state of 
$v_{\rm gen}$ in $C$ at time $0$.)

Moreover, let $A''$ be the finite automaton 
of the traditional model 
that sees the boundary condition $\mathbf{b} = {\rm bc}_{C}(v_{\rm gen})$ 
from time $0$ to time $1$ 
and at time $1$ broadcasts $\mathbf{b}$ from $v_{\rm gen}$ 
to all nodes as a message.

Finally let $A'$ be the finite automaton of the traditional model 
that simulates all of $A_{\mathbf{b}}$ ($\mathbf{b} \in \{0, 1\}^{4}$) 
and $A''$ and fires at a time $t$ if there is $\mathbf{b}$ 
such that $A_{\mathbf{b}}$ fires at the time $t$ and 
the node has received the message $\mathbf{b}$ before or at time $t$.

We can easily show that $A'$ is a desired partial solution of $\Gamma$ 
of the traditional model.
\end{proof}

\mynewpage

The two variations $\text{\rm g-LSP}[a,b]$ and 
$\text{\rm g-LSP}[b,a]$ are essentially the same problem.
Therefore we determine the value 
${\rm mft}_{\text{\rm g-LSP}[a,b]}(C)$ only for 
$a \leq b$.

\begin{thm}
\label{theorem:mft_g_lsp_a_b}
Suppose that $a \leq b$.
Then 
\[
{\rm mft}_{\text{\rm g-LSP}[a,b]}(C_{\rm L}(w, h, i)) = 
\begin{cases}
w + h & \text{if $0 \leq i < 2w$},\\
i - w + h & \text{if $2w \leq i \leq w + h$}.
\end{cases}
\]
\end{thm}

\begin{proof}
Let $C(w_{0}, h_{0}, i_{0})$ be some fixed configuration 
of $\Gamma = \text{\rm g-LSP}[a, b]$. 
We determine the value ${\rm mft}(C(w_{0}, h_{0}, i_{0}))$.
First we consider the case of the boundary sensitive model 
and then consider that of the traditional model.

\medskip

\noindent
(I) The case of the boundary sensitive model and 
$0 \leq i_{0} < w_{0}$.

\medskip

We construct a check and broadcast partial solution $A$ 
to prove the upper bound 
${\rm mft}(C(w_{0}, h_{0}, i_{0})) \leq w_{0} + h_{0}$.
Let $x_{0}, y_{0}$ be the values 
$x_{0} = i_{0}$, $y_{0} = w_{0} - x_{0}$.
Suppose that copies of $A$ are placed in $C(w, h, i)$.
$A$ uses two signals ${\rm R}_{0}$, ${\rm R}_{1}$.

At time $0$, both of ${\rm R}_{0}$, ${\rm R}_{1}$ 
check the boundary condition of $v_{\rm gen}$ 
(that is, the statement ${\rm bc}_{C(w, h, i)}(p_{i}) 
= {\rm bc}_{C(w_{0}, h_{0}, i_{0})}(p_{i_{0}})$).
If the boundary condition is correct then we have 
$0 \leq i < w$ and let $x$, $y$ be the values 
$x = i$, $y = w - i$.
If the boundary condition of $v_{\rm gen}$ is correct, 
${\rm R}_{0}$ proceeds to the west and checks 
the statement $x = x_{0}$.
If the statement is true, the signal broadcasts 
the message ``$x = x_{0}$'' and then vanishes.
If the boundary condition of $v_{\rm gen}$ is correct, 
${\rm R}_{1}$ proceeds to the east and checks 
the statement $y = y_{0}$.
If the statement is true, 
the signal broadcasts the message ``$y = y_{0}$'', 
proceeds to the north and checks the statement $h = h_{0}$.
If the statement is true, the signal broadcasts 
the message ``$h = h_{0}$'' and then vanishes.

A node fires at time $w_{0} + h_{0}$ if it has 
received at least two of the three messages 
``$x = x_{0}$'', ``$y = y_{0}$'', ``$h = h_{0}$'' 
before or at time $w_{0} + h_{0}$.

Suppose that $C(w, h, i) = C(w_{0}, h_{0}, i_{0})$.
Then the boundary condition of $v_{\rm gen}$ is correct and 
all of the three messages 
``$x = x_{0}$'', ``$y = y_{0}$'', ``$h = h_{0}$'' 
are generated.
In Figure \ref{figure:fig010} we show a diagram 
that explains the travels of the two signals ${\rm R}_{0}$, 
${\rm R}_{1}$ (lines labeled with ``${\rm R}_{0}$'', 
``${\rm R}_{1}$'') and 
the propagations of the three messages 
(lines labeled with the names of messages 
such as ``$x = x_{0}$'').
The generations of messages are marked with black circles.
The messages ``$x = x_{0}$'', ``$y = y_{0}$'', ``$h = h_{0}$'' 
are generated at $p_{0}$, $p_{w_{0}}$, $p_{w_{0} + h_{0}}$ 
at times $x_{0}$, $y_{0}$, $y_{0} + h_{0}$ respectively.
The thick line denotes the first time when a node receives 
at least two of the three messages.
The diagram shows that any node in $C(w, h, i)$ receives 
two messages before or at time $w_{0} + h_{0}$.
Hence all nodes in $C(w, h, i)$ fire at $w_{0} + h_{0}$.

\begin{figure}[htb]
\begin{center}
\includegraphics[scale=1.0]{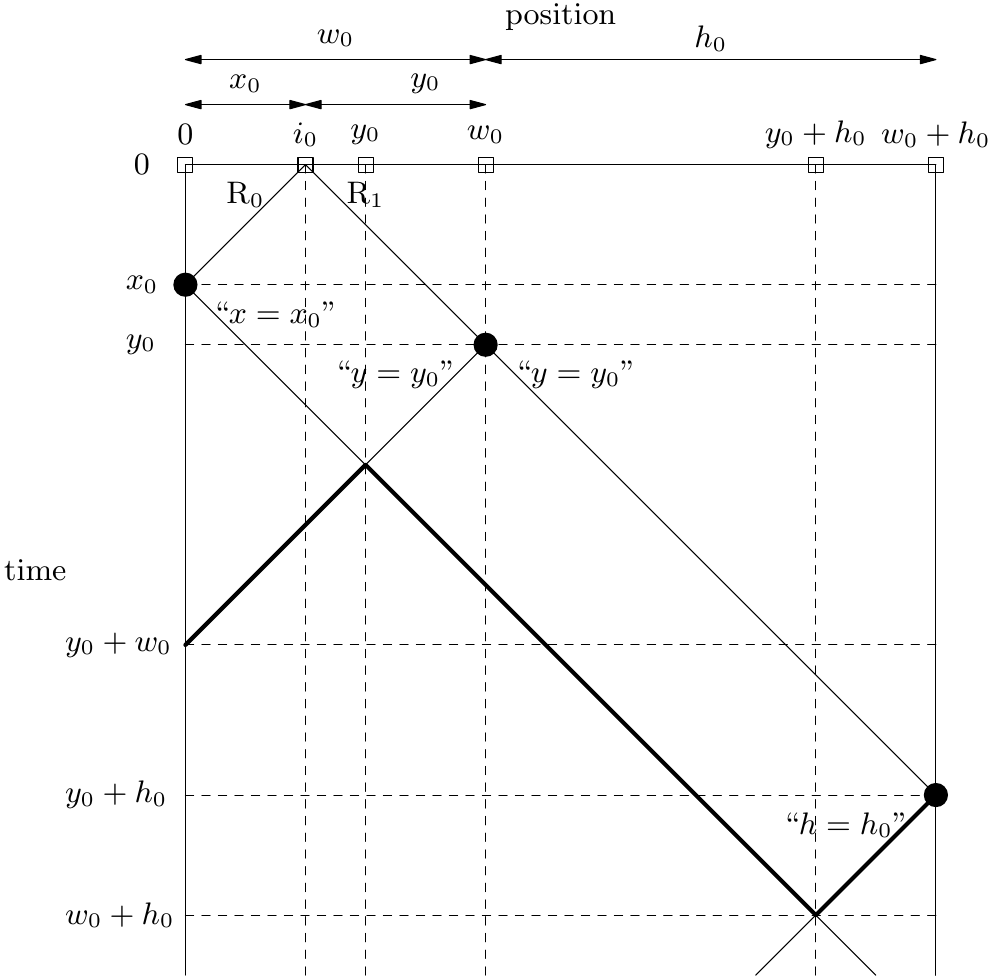}
\end{center}
\caption{The diagram used in the proof 
of Theorem \ref{theorem:mft_g_lsp_a_b} for the case 
$0 \leq i_{0} < w_{0}$.}
\label{figure:fig010}
\end{figure}

Conversely suppose that a node in $C(w, h, i)$ fires at some time.
This means that the node receives at least two of the three messages 
and hence at least two of the three statements 
$x = x_{0}$, $y = y_{0}$, $h = h_{0}$ are true.
But from any two of these three statements follows that 
$C(w, h, i) = C(w_{0}, h_{0}, i_{0})$ because 
$x + y = w$, $x_{0} + y_{0} = w_{0}$, $w / h = w_{0} / h_{0} = a / b$, 
$i = x$, $i_{0} = x_{0}$.

This shows that $A$ is a partial solution that has 
the domain $\{C(w_{0}, h_{0}, i_{0})\}$ and that fires 
$C(w_{0}, h_{0}, i_{0})$ at time $w_{0} + h_{0}$.
From this we have the upper bound 
${\rm mft}(C(w_{0}, h_{0}, i_{0})) \leq w_{0} + h_{0}$.

\mynewpage

Next we prove the lower bound 
$w_{0} + h_{0} \leq {\rm mft}(C(w_{0}, h_{0}, i_{0}))$ 
using Lemma \ref{lemma:fact_4}.
Let $C$ be $C(w_{0}, h_{0}, i_{0})$.
By Figure \ref{figure:fig010} we know 
that at time $(w_{0} + h_{0}) - 1$ the node 
$p_{y_{0} + h_{0}}$ knows the value of $y$ ($= y_{0}$) but 
none of the values of $x, h$.
Hence if we select a configuration $C(w, h, i)$ such 
that $y = w - i = y_{0}$ but $x, h$ are sufficiently large 
as $C'$ 
we can prove the lower bound by 
Lemma \ref{lemma:fact_4}.

We show this using the configuration $C = C(4, 6, 2)$ 
of $\text{\rm g-LSP}[2, 3]$ as an example.
For this case $w_{0} = 4$, $h_{0} = 6$, $i_{0} = 2$, 
$x_{0} = 2$, $y_{0} = 2$, $y_{0} + h_{0} = 8$, 
and $w_{0} + h_{0} = 10$.
We select the configuration $C(6, 9, 4)$ as $C'$.
In Figure \ref{figure:fig051} (a) and (b) 
we show $C$ and $C'$.
For these configurations we have 
${\rm ai}(p_{8}, 9, C) = {\rm ai}(p_{10}, 9, C')$.
In (c) we show their local map components.
$C'$ contains two nodes $p_{14}$, $p_{15}$ 
(marked with ``*'' in (b)) 
that can be used to prove 
${\rm rad}(C') > 9$.
Therefore we have 
$9 < {\rm mft}(C) = {\rm mft}(C(4, 6, 2))$.

\begin{figure}[htb]
\begin{center}
\includegraphics[scale=1.0]{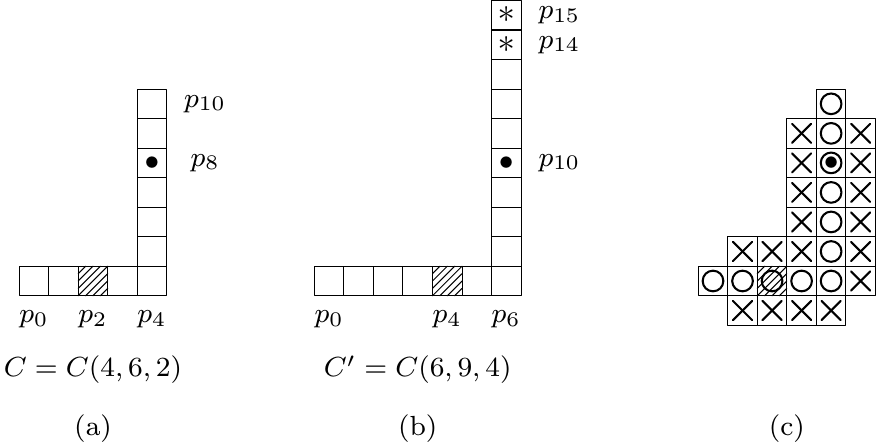}
\end{center}
\caption{(a) and (b) show 
$C = C(4, 6, 2)$ and $C' = C(6, 9, 4)$ 
respectively.
(c) shows the local map component of 
${\rm ai}(p_{8}, 9, C) = {\rm ai}(p_{10}, 9, C')$.}
\label{figure:fig051}
\end{figure}

\mynewpage

\medskip

\noindent
(II) The case of the boundary sensitive model and 
$w_{0} \leq i_{0} \leq w_{0} + h_{0}$.

\medskip

We explain only how to modify the proof for the previous case.
Let $x_{0}$, $y_{0}$ be the values 
$x_{0} = i_{0} - w_{0}$, $y_{0} = w_{0} + h_{0} - i_{0}$.
If the boundary condition of $v_{\rm gen}$ is correct 
we have $w \leq i \leq w + h$.
Let $x$, $y$ be the values 
$x = i - w$, $y = w + h - i$.

\begin{figure}[htbp]
\begin{center}
\includegraphics[scale=1.0]{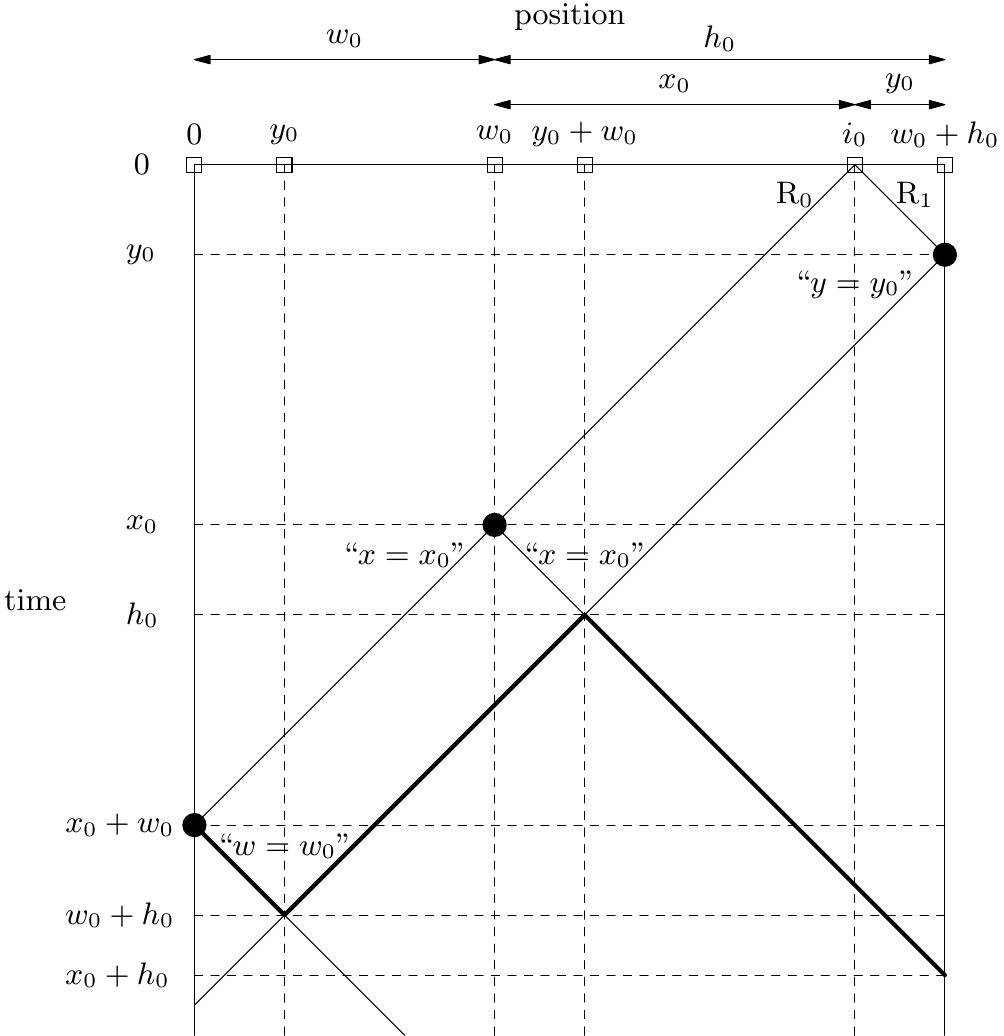}
\end{center}
\caption{The diagram used in the proof of 
Theorem \ref{theorem:mft_g_lsp_a_b} for the case 
$w_{0} \leq i_{0} \leq w_{0} + h_{0}$.}
\label{figure:fig011}
\end{figure}

In Figure \ref{figure:fig011} we show the diagram of 
signals and messages for this case. 
The signal ${\rm R}_{0}$ generates the message 
``$x = x_{0}$'' and also the message ``$w = w_{0}$''.
The signal ${\rm R}_{1}$ generates only the message 
``$y = y_{0}$''.
A node fires at time $t_{0}$ if it has received 
at least two of the three messages 
``$x = x_{0}$'', ``$y = y_{0}$'', ``$w = w_{0}$'' 
before or at time $t_{0}$.
Here $t_{0}$ is the smallest time when all nodes 
have received the necessary two messages.

By the diagram we know that 
there are two candidates for a node that receives 
the necessary two messages last.
One is $p_{y_{0}}$ and the receiving time is 
$w_{0} + h_{0}$.
The other is $p_{w_{0} + h_{0}}$ and 
the receiving time is $x_{0} + h_{0}$.
Hence, if $x_{0} < w_{0}$ (and hence 
$i_{0} < 2w_{0}$), $x_{0} + h_{0} < w_{0} + h_{0}$ 
and we determine the firing time $t_{0}$ to be 
$t_{0} = w_{0} + h_{0}$.
If $w_{0} \leq x_{0}$ (and hence 
$2w_{0} \leq i_{0}$), 
$w_{0} + h_{0} \leq x_{0} + h_{0}$ and we determine $t_{0}$ to be 
$t_{0} = x_{0} + h_{0} = i_{0} - w_{0} + h_{0}$.
This gives the desired upper bounds 
$w_{0} + h_{0}$ for $w_{0} \leq i_{0} < 2w_{0}$ 
and $i_{0} - w_{0} + h_{0}$ 
for $2w_{0} \leq i_{0} \leq w_{0} + h_{0}$.

The proof of the lower bounds using 
Lemma \ref{lemma:fact_4} is similar.

\mynewpage

\medskip

\noindent
(III) The case of the traditional model.

\medskip

Finally we consider the case of the traditional model 
and prove \\
${\rm mft}_{\Gamma, {\rm bs}}(C(w, h, i)) =  
{\rm mft}_{\Gamma, {\rm tr}}(C(w, h, i))$ for any 
$C(w, h, i)$.
The proofs for lower bounds in (I), (II) are true 
also for the traditional model because 
Lemma \ref{lemma:fact_4} is true also for the 
traditional model.
We consider upper bounds.

In the previous proof for the boundary sensitive model, 
for any $C(w, h, i)$ we constructed a partial 
solution $A$ of the boundary sensitive model 
that has the domain 
$\{C(w, h, i)\}$ and that fires $C(w, h, i)$ at time 
${\rm mft}_{\Gamma, {\rm bs}}(C(w, h, i))$.
Moreover we can show that 
except two cases that we explain later, we have 
${\rm mft}_{\Gamma, {\rm bs}}(C(w, h, i)) > 
\max\{i, w + h - i\} = {\rm rad}(C(w, h, i))$. 
Hence, by Lemma \ref{lemma:from_bs_to_tr}, 
except the two cases mentioned above we have 
${\rm mft}_{\Gamma, {\rm bs}}(C(w, h, i))$ $=$ 
${\rm mft}_{\Gamma, {\rm tr}}(C(w$, $h$, $i))$.

The first exception is configurations of 
the form $C(w, h, 0)$. 
The second exception is the case where 
$a = b$ and the configuration is 
of the form $C(w, w, 2w)$.
The proof for the second exception is similar 
to that for the first. 
Therefore we consider only the first exception.

\mynewpage

Let $C(w_{0}, h_{0}, 0)$ be some fixed configuration 
of $\Gamma$.
We construct a partial solution $A'$ of the traditional model 
that has the domain $\{ C(w_{0}, h_{0}, 0) \}$ and fires 
$C(w_{0}, h_{0}, 0)$ at time $w_{0} + h_{0}$.

$A'$ uses three signals ${\rm R}_{0}$, ${\rm R}_{1}$, ${\rm R}_{2}$.
Signal ${\rm R}_{0}$ checks the statement $i = 0$ 
by checking the equivalent statement 
${\rm bc}_{C}(v_{\rm gen}) = (1, 0, 0, 0)$.
If it is true then the signal broadcasts the message 
``yes-0'' at $p_{0}$ at time $1$.

Signal ${\rm R}_{1}$ checks the statement 
$w - i = w_{0}$ and broadcasts the message 
``yes-1'' if it is true.
Note that for ${\rm R}_{1}$ to check this statement 
it is not necessary to know ${\rm bc}_{C}(v_{\rm gen})$.
If the statement is true then 
the signal generates the message ``yes-1'' 
at $p_{w}$ at time $w_{0}$.

Signal ${\rm R}_{2}$ checks the statement 
``$w - i = w_{0}$ and $h = h_{0}$'' and broadcasts 
the message ``yes-2'' if 
the statement is true.
Similarly to the message ``yes-1'', 
the signal generates the message ``yes-2'' 
at $p_{w + h_{0}}$ at time 
$w_{0} + h_{0}$ if the statement is true.

A node fires at time $w_{0} + h_{0}$ if either it has 
received both of the two messages ``yes-0'' and ``yes-1'' 
or it has received the message ``yes-2'' before or 
at time $w_{0} + h_{0}$.

Suppose that $C(w, h, i) = C(w_{0}, h_{0}, 0)$.
Then all of $i = 0$, $w - i = w_{0}$, $h = h_{0}$ are true and hence 
``yes-0'' is generated at $p_{0}$ at time $1$, 
``yes-1'' is generated at $p_{w_{0}}$ at time $w_{0}$, and 
``yes-2'' is generated at $p_{w_{0}+h_{0}}$ at time $w_{0} + h_{0}$.
Therefore $p_{j}$ such that $0 \leq j < w_{0} + h_{0}$ 
fires at $w_{0} + h_{0}$ by receiving ``yes-0'' and ``yes-1'' 
and $p_{w_{0} + h_{0}}$ fires at $w_{0} + h_{0}$ by 
receiving ``yes-2''.

Conversely suppose that a node in $C(w, h, i)$ fires 
at some time.
If it fires by receiving ``yes-0'' and ``yes-1'' 
then $i = 0$ and $w - i = w_{0}$ are true and hence 
$w = w_{0}$, $h = h_{0}$ are also true because 
$w / h = w_{0} / h_{0} = a / b$.
Similarly if the node fires by receiving ``yes-2'' then 
both of $w - i = w_{0}$, $h = h_{0}$ are true and hence 
$i = 0$, $w = w_{0}$ are also true.
In each case we have $C(w, h, i) = C(w_{0}, h_{0}, 0)$.

Therefore, $A'$ is a partial solution of $\Gamma$ of the 
traditional model with the domain 
$\{C(w_{0}, h_{0}, 0)\}$ that fires $C(w_{0}, h_{0}, 0)$ at 
time $w_{0} + h_{0}$.
This proves the upper bound 
${\rm mft}_{\Gamma, {\rm tr}}(C(w_{0}, h_{0}, 0)) 
\leq w_{0} + h_{0}$.
\end{proof}

\mynewpage

To prove nonexistence of minimal-time solutions of 
$\text{\rm g-LSP}[a,b]$ we use one very simple 
proof technique.

Let $\Gamma$, $\Gamma'$ be two variations.
We say that $\Gamma'$ is a {\it supervariation} 
of $\Gamma$ 
if any configuration $C$ of $\Gamma$ is 
also a configuration of $\Gamma'$.
In this case any solution of $\Gamma'$ is also 
a solution of $\Gamma$ and hence 
we have ${\rm mft}_{\Gamma}(C) \leq 
{\rm mft}_{\Gamma'}(C)$ for any configuration $C$ of 
$\Gamma$.

Moreover, we say that 
$\Gamma'$ is a {\it conservative supervariation} 
of $\Gamma$ 
if any configuration $C$ of $\Gamma$ is also 
a configuration of $\Gamma'$ and 
${\rm mft}_{\Gamma}(C) = {\rm mft}_{\Gamma'}(C)$.
In this case any minimal-time solution of $\Gamma'$ is also 
a minimal-time solution of $\Gamma$.
Hence if $\Gamma$ has no minimal-time solutions 
then $\Gamma'$ also has no minimal-time solutions.

\begin{thm}
\label{theorem:non_existence_g_lsp_a_b}
For any $a, b$, 
$\text{\rm g-LSP}[a,b]$ has no minimal-time solutions.
\end{thm}

\begin{proof}
Consider the case $a \leq b$.
Let $C = C(w, h, 0)$ be a configuration of both of 
$\Gamma = \text{\rm LSP}[a,b]$ and 
$\Gamma' = \text{\rm g-LSP}[a,b]$.
Then by Theorem \ref{theorem:mft_lsp_a_b} 
we have ${\rm mft}_{\Gamma}(C) = w + h$ 
and by Theorem \ref{theorem:mft_g_lsp_a_b} 
we have ${\rm mft}_{\Gamma'}(C) = w + h$.
Therefore $\Gamma'$ is a conservative supervariation 
of $\Gamma$.
However, by Theorem \ref{theorem:yamashita_result_2} 
$\Gamma$ has no minimal-time solutions.
Therefore, $\Gamma'$ has also no minimal-time solutions.

Next consider the case $a > b$.
By the above proof we know that $\text{\rm g-LSP}[b,a]$ 
has no minimal-time solutions.
However the two variations 
$\text{\rm g-LSP}[b,a]$ and $\text{\rm g-LSP}[a,b]$ 
are essentially the same problem.
Hence $\text{\rm g-LSP}[a,b]$ has no minimal-time 
solutions.
\end{proof}

\mynewpage

\section{Variations of FSSP for rectangular walls}
\label{section:rec_walls}

In this section we consider the four variations of FSSP 
for rectangular walls, 
that is, $\text{\rm RECT-WALL}$, $\text{\rm RECT-WALL}[a, b]$ 
(see Table \ref{table:fig027}) and their 
generalized variations 
$\text{\rm g-RECT-WALL}$, $\text{\rm g-RECT-WALL}[a, b]$.

As for $\text{\rm RECT-WALL}$ and $\text{\rm g-RECT-WALL}$, 
at present we do not know whether they have 
minimal-time solutions or not.
The FSSP's for one-way and two-way rings have been 
extensively studied and we know their minimal-time solutions 
(\cite{
Berthiaume_2004, 
Culik_1987, 
Kobayashi_Res_Rep_1976,
LaTorre_1996}).
$\text{\rm RECT-WALL}$ and $\text{\rm g-RECT-WALL}$ 
are natural modifications of the FSSP of 
two-way rings.
It is a very interesting open problem to 
determine whether they have minimal-time solutions or not.

All configurations of the four variations 
are rectangular walls 
$C_{\rm RW}(w, h)$ (see Figure \ref{figure:fig001} (b)).
We assume that 
$p_{-w-h}, p_{-w-h+1}, \ldots, p_{-h}$, $\ldots$, 
$p_{-1}$, $p_{0}$, $p_{1}$, $\ldots$, $p_{w}$, $\ldots$, 
$p_{w+h-1}, p_{w+h}$ is the sequence 
of nodes in $C_{\rm RW}(w, h)$ obtained by tracing it 
from the northeast corner ($p_{-w-h}$) to 
the northwest corner ($p_{-h}$), to the southwest corner ($p_{0}$), 
to the southeast corner ($p_{w}$), 
and to the northeast corner ($p_{w+h}$) again.
Moreover we understand the index $i$ of $p_{i}$ ``modulo $2w + 2h$'' 
so that, for example, 
the first node $p_{-w-h}$ and 
the last node $p_{w+h}$ in the above sequence 
are the same node (the northeast corner).
When we consider $C_{\rm RW}(w, h)$ as a configuration of 
generalized variations such that $p_{i}$ is the position of the general 
we denote it by $C_{\rm RW}(w, h, i)$.

\mynewpage

First we derive the value of the minimum firing time of 
$\text{\rm RECT-WALL}$.

\begin{thm}
\label{theorem:mft_rect_wall}
\[
{\rm mft}_\text{\rm RECT-WALL}(C_{\rm RW}(w, h)) = 
w + h + \max\{w, h\}.
\]
\end{thm}

\begin{proof}
Let $C(w_{0}, h_{0})$ be some fixed configuration of 
$\text{\rm RECT-WALL}$.
We prove 
${\rm mft}(C(w_{0}, h_{0})) = t_{0}$, 
where $t_{0} = w_{0} + h_{0} + \max \{w_{0}, h_{0}\}$.

First we prove the upper bound 
${\rm mft}(C(w_{0}, h_{0})) \leq t_{0}$ by 
constructing a check and broadcast partial solution $A$.
Suppose that copies of $A$ are placed in a configuration 
$C(w, h)$ of $\text{\rm RECT-WALL}$.
$A$ uses two signals ${\rm R}_{0}$, ${\rm R}_{1}$.

The signal ${\rm R}_{0}$ proceeds from the general 
to the north, checks the statement $h = h_{0}$, 
proceeds to the east, and checks the statement 
$w = w_{0}$.
The signal ${\rm R}_{1}$ proceeds from the general 
to the east, checks the statement $w = w_{0}$, 
proceeds to the north, and checks the statement 
$h = h_{0}$.
In any case, if a signal finds that a statement is 
true then it broadcasts a message such as 
``$w = w_{0}$'', ``$h = h_{0}$''.
A node fires at time $t_{0}$ if it has received 
both of the two messages ``$w = w_{0}$'', ``$h = h_{0}$'' 
before or at time $t_{0}$.

Suppose that $C(w, h) = C(w_{0}, h_{0})$.
In Figure \ref{figure:fig013} we show a diagram 
that shows the travel of signals and 
the generation and the propagation of messages.
In this figure we represent $C(w_{0}, h_{0})$ by 
a line of positions 
$p_{-w_{0}-h_{0}}, \ldots$, $p_{-h_{0}}, \ldots$, $p_{0}, 
\ldots, p_{w_{0}}, \ldots$, $p_{w_{0}+h_{0}}$. 
The leftmost position $p_{-w_{0}-h_{0}}$ 
and the rightmost position $p_{w_{0}+h_{0}}$ 
are the same position. 
As is in Figure \ref{figure:fig010}, 
the generation of a message is marked with a black 
circle.

\begin{figure}[htb]
\begin{center}
\includegraphics[scale=1.0]{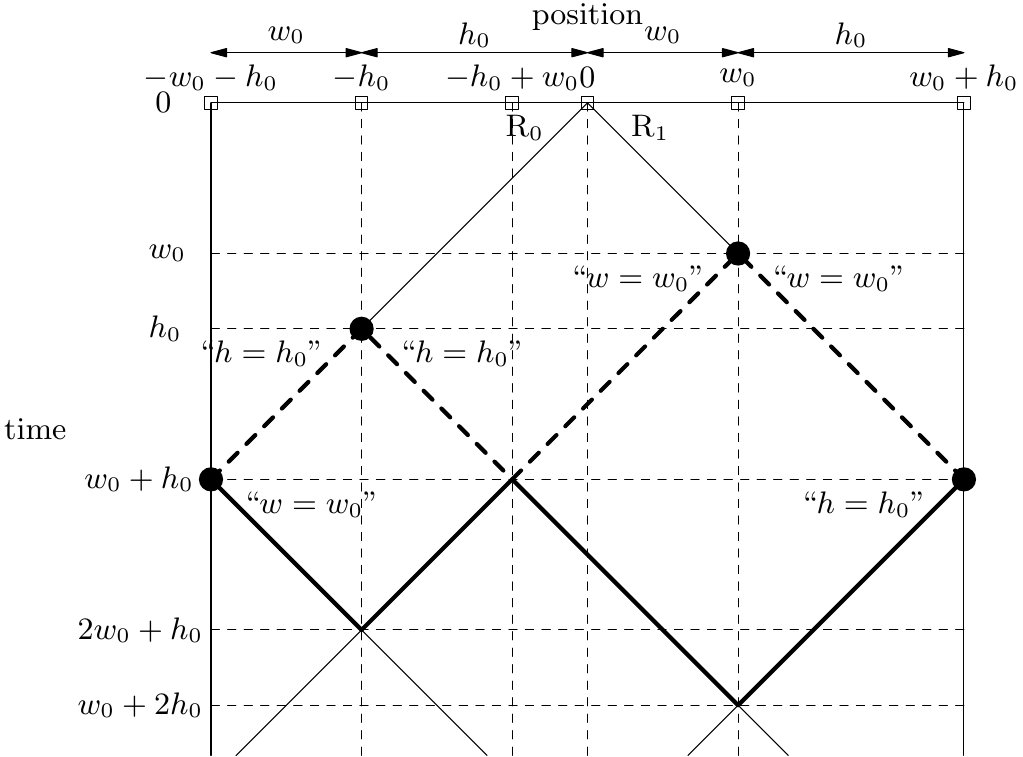}
\end{center}
\caption{The diagram used in the proofs of 
Theorem \ref{theorem:mft_rect_wall} and 
Theorem \ref{theorem:mft_rect_wall_a_b}.}
\label{figure:fig013}
\end{figure}
The thick line represents the first time 
a node receives both of the two messages.
This line shows that all nodes 
receive both of the two message before or at 
time $t_{0} = \max \{2w_{0} + h_{0}, w_{0} + 2h_{0}\}$.
Therefore, all nodes in $C(w, h)$ fire at time $t_{0}$.

Conversely suppose that a node fires 
at some time.
This means that both of the two statements 
$w = w_{0}$, $h = h_{0}$ are true and hence 
$C(w, h) = C(w_{0}, h_{0})$.

Therefore, $A$ is a partial solution 
that has the domain $\{C(w_{0}, h_{0})\}$ 
and that fires $C(w_{0}, h_{0})$ at time $t_{0}$.
This shows the desired upper bound 
${\rm mft}(C(w_{0}, h_{0}))$ $\leq$ $t_{0}$.

We can show the lower bound 
$t_{0} \leq {\rm mft}(C(w_{0}, h_{0}))$ 
by Lemma \ref{lemma:fact_4} using the fact that 
at time $t_{0} - 1$ 
either the node $p_{-h_{0}}$ cannot know 
the value of $w_{0}$ (for the case $h_{0} \leq w_{0}$) 
or the node $p_{w_{0}}$ cannot know 
the value of $h_{0}$ (for the case $w_{0} \leq h_{0}$).
For more details, see the proof of the lower bound 
in the proof of Theorem \ref{theorem:mft_g_lsp_a_b} 
and Figure \ref{figure:fig051}.
\end{proof}

\mynewpage

By slightly modifying the proof of 
Theorem \ref{theorem:mft_rect_wall} we can derive the value 
${\rm mft}_{\text{\rm RECT-WALL}[a,b]}(C(w, h))$.

\begin{thm}
\label{theorem:mft_rect_wall_a_b}
\[
{\rm mft}_{\text{\rm RECT-WALL}[a,b]}(C_{\rm RW}(w, h)) = 
w + h.
\]
\end{thm}

\begin{proof}
In the proof of Theorem \ref{theorem:mft_rect_wall} 
we designed $A$ so that a node fires at time 
$t_{0} = w_{0} + h_{0} + \max\{ w_{0}, h_{0}\}$
if it has received both of the two 
messages ``$w = w_{0}$'', ``$h = h_{0}$'' before or at 
time $t_{0}$.
We change this rule so that 
a node fires at time $w_{0} + h_{0}$ if 
it has received at least one of the two messages before or at 
time $w_{0} + h_{0}$.

In Figure \ref{figure:fig013}, by a dotted thick line we show 
the first time a node receives at least one of the two messages.
From this line we know that if $C(w, h) = C(w_{0}, h_{0})$ 
all nodes receive at least one of the two messages before or at 
time $w_{0} + h_{0}$ and hence fire at that time.
Conversely, if a node in $C(w, h)$ fires at a time it means 
that at least one of $w = w_{0}$, $h = h_{0}$ is true and hence 
$C(w, h) = C(w_{0}, h_{0})$ (because $w / h = w_{0} / h_{0} = a / b$).
Therefore $A$ is a partial solution that has the domain 
$\{C(w_{0}, h_{0})\}$ and fires $C(w_{0}, h_{0})$ at time 
$w_{0} + h_{0}$.

The lower bound can be proved by 
Lemma \ref{lemma:fact_4} 
using the fact that both of $p_{-w_{0} - h_{0}}$ and $p_{-h_{0} + w_{0}}$ 
can know neither $w_{0}$ nor $h_{0}$ at time $w_{0} + h_{0} - 1$.
\end{proof}

\mynewpage

In determining the value 
${\rm mft}(C(w, h, i))$ for $\text{\rm g-RECT-WALL}$ 
we may assume that $0 \leq w \leq h$ and $0 \leq i \leq w + h$ 
because $\text{\rm g-RECT-WALL}$ has many symmetries in it.

\begin{thm}
\label{theorem:mft_g_rect_wall}
Suppose that $0 \leq w \leq h$.
\begin{enumerate}
\item[{\rm (1)}] If $0 \leq i \leq w$ then 
${\rm mft}_\text{\rm g-RECT-WALL}(C_{\rm RW}(w, h, i)) = w + 2h$.
\item[{\rm (2)}] If $w < i \leq w + h$ and $h \leq 2w$ then
\begin{eqnarray*}
\lefteqn{{\rm mft}_\text{\rm g-RECT-WALL}(C_{\rm RW}(w, h, i))} \\
& = & 
\begin{cases}
-i + 2w + 2h & \text{if \/\/ $w < i \leq h$},\\
2w + h & \text{if \/\/ $h < i \leq 2w$},\\
i + h & \text{if \/\/ $2w < i \leq w + h$}.
\end{cases}
\end{eqnarray*}
\item[{\rm (3)}] If $w < i \leq w + h$ and $2w < h$ then
\begin{eqnarray*}
\lefteqn{{\rm mft}_\text{\rm g-RECT-WALL}(C_{\rm RW}(w, h, i))} \\
& = &
\begin{cases}
-i + 2w + 2h & \text{if \/\/ $w < i \leq w + h/2$}, \\
i + h & \text{if \/\/ $w + h/2 < i \leq w + h$}.
\end{cases}
\end{eqnarray*}
\end{enumerate}
\end{thm}

\begin{proof}
The proof of the lower bound using Lemma 
\ref{lemma:fact_4}
is similar to those of 
Theorems 
\ref{theorem:mft_rect_wall}.
\ref{theorem:mft_rect_wall_a_b}.
Therefore we consider only the upper bound.
We explain only how to modify the construction of 
the check and broadcast partial solution $A$ 
in the proof of 
Theorem \ref{theorem:mft_rect_wall}.
Let $C(w_{0}, h_{0}, i_{0})$ be some fixed configuration 
of $\text{\rm g-RECT-WALL}$.
We assume $w_{0} \leq h_{0}$.
Suppose that copies of $A$ are placed in a configuration 
$C(w, h, i)$ of $\text{\rm g-RECT-WALL}$.
$A$ uses two signals ${\rm R}_{0}$, ${\rm R}_{1}$.

\medskip

\noindent
(I) The case of the boundary sensitive model 
and $0 \leq i_{0} \leq w_{0}$.

\medskip

Let $x_{0} = i_{0}$, $y_{0} = w_{0} - i_{0}$.
The signal ${\rm R}_{0}$ checks the boundary condition of 
the general.
If it is correct then $0 \leq i \leq w$.
Let $x$, $y$ be values 
$x = i$, $y = w - i$.
The signal proceeds to the west, checks 
the statement $x = x_{0}$, proceeds to the north, 
and checks the condition $h = h_{0}$.
The signal ${\rm R}_{1}$ checks 
the boundary condition of the general, 
proceeds to the east, checks the condition $y = y_{0}$, 
proceeds to the north, and checks the condition 
$h = h_{0}$.
If a statement is true then the signal broadcasts 
a message with the name of the statement as its label.
A node fires at time $w_{0} + 2h_{0}$ if it 
has received all of the three messages 
``$x = x_{0}$'', ``$y = y_{0}$'', ``$h = h_{0}$'' 
before or at time $w_{0} + 2h_{0}$.

\begin{figure}[htb]
\begin{center}
\includegraphics[scale=1.0]{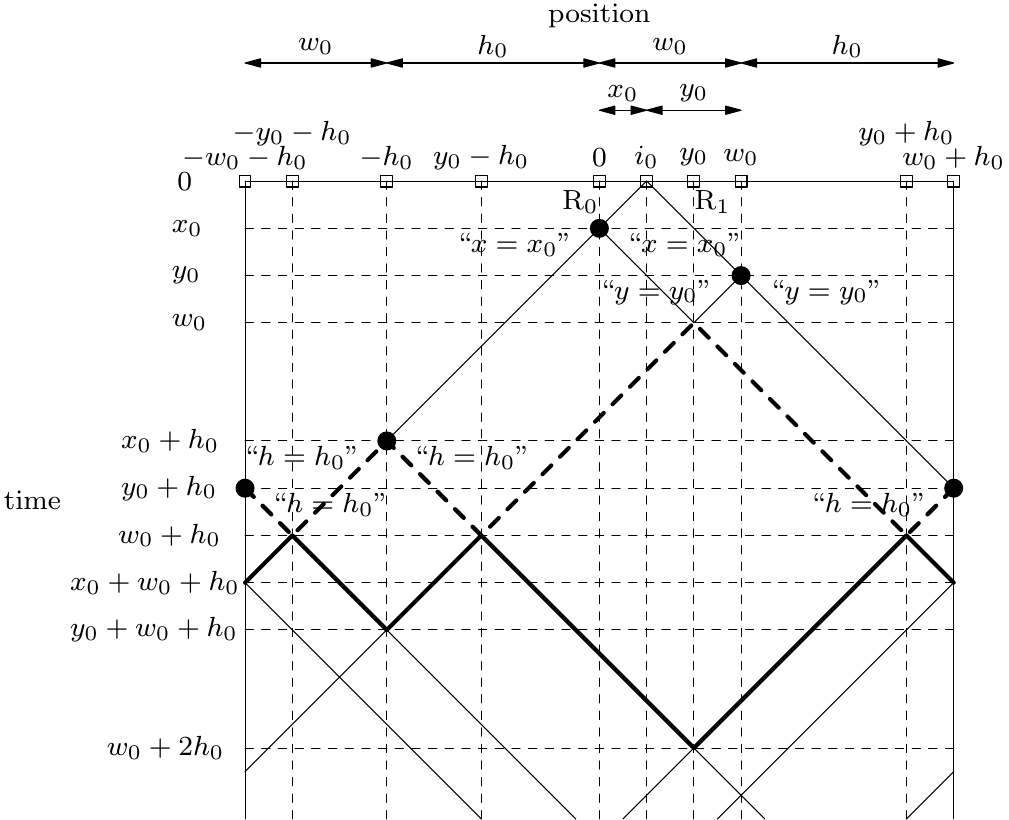}
\end{center}
\caption{The diagram used in the proofs of 
Theorems \ref{theorem:mft_g_rect_wall} and 
\ref{theorem:mft_g_rect_wall_a_b} 
for the case $0 \leq i_{0} \leq w_{0}$.}
\label{figure:fig014}
\end{figure}

Suppose that $C(w, h, i) = C(w_{0}, h_{0}, i_{0})$.
In Figure \ref{figure:fig014} we show 
the travel of signals and the generation and the propagation 
of messages.
The thick line represents the first time a node 
receives all of the three messages.
The line shows that all nodes receive 
all of the three messages before or at 
time $y_{0} + 2h_{0}$ and fire 
at the time 
(note that $x_{0} + w_{0} + h_{0} \leq w_{0} + 2h_{0}$ 
and $y_{0} + w_{0} + h_{0} \leq w_{0} + 2h_{0}$).

Conversely suppose that 
a node in $C(w, h, i)$ fires at some time.
This means that all of $x = x_{0}$, $y = y_{0}$, $h = h_{0}$ 
are true and hence $C(w, h, i) = C(w_{0}, h_{0}, i_{0})$.

Therefore, $A$ is a partial solution that has 
the domain $\{C(w_{0}, h_{0}, i_{0})\}$ and that 
fires $C(w_{0}, h_{0}, i_{0})$ at time 
$w_{0} + 2h_{0}$.
This shows the desired upper bound.

\medskip

\mynewpage

\noindent
(II) The case of the boundary sensitive model and 
$w_{0} < i_{0} \leq w_{0} + h_{0}$.

\medskip

In this case, let $x_{0}$, $y_{0}$, $x$, $y$ be 
the values 
$x_{0} = i_{0} - w_{0}$, $y_{0} = w_{0} + h_{0} - i_{0}$, 
$x = i - w$, $y = w + h - i$.
The signals ${\rm R}_{0}$, ${\rm R}_{1}$ check 
the statements $x = x_{0}$, $y = y_{0}$, $w = w_{0}$ 
and broadcast messages if statements are true.
A node fires at time 
$t_{0} = \max\{x_{0}, y_{0}, w_{0}\} + w_{0} + h_{0}$ 
if it has received all of the three messages 
``$x = x_{0}$'', ``$y = y_{0}$'', ``$w = w_{0}$'' before 
or at time $t_{0}$.

\begin{figure}[htb]
\begin{center}
\includegraphics[scale=1.0]{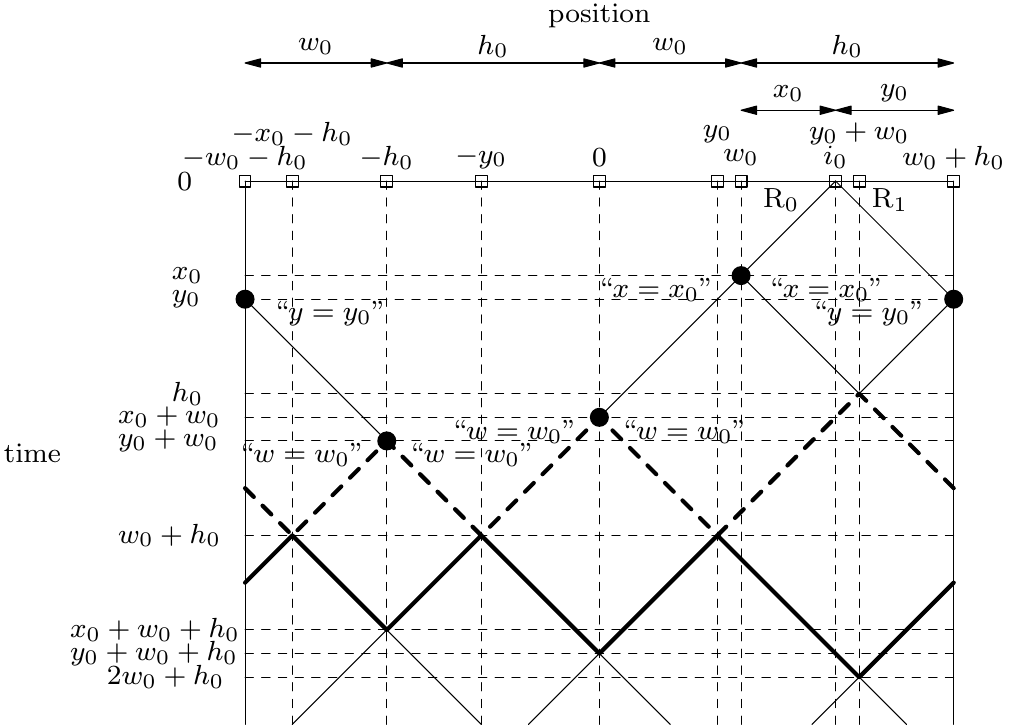}
\end{center}
\caption{The diagram used in the proofs of 
Theorems \ref{theorem:mft_g_rect_wall} and 
\ref{theorem:mft_g_rect_wall_a_b} 
for the case $w_{0} < i_{0} \leq w_{0} + h_{0}$.}
\label{figure:fig015}
\end{figure}

Suppose that $C(w, h, i) = C(w_{0}, h_{0}, i_{0})$.
In Figure \ref{figure:fig015} the thick line represents 
the first time a node receives all of the three messages.
It shows that all nodes in $C(w, h, i)$ 
receive all of the three 
messages before or at time 
$t_{0} = \max \{x_{0} + w_{0} + h_{0}, y_{0} + w_{0} + h_{0}, 
2w_{0} + h_{0}\}$ and hence fire at that time.

Conversely suppose that a node in $C(w, h, i)$ fires 
at some time.
This means that all of $x=x_{0}$, $y = y_{0}$, $w = w_{0}$ are 
true and $C(w, h, i) = C(w_{0}, h_{0}, i_{0})$.

Therefore, we have an upper bound 
${\rm mft}(C(w_{0}, h_{0}, i_{0})) \leq t_{0}$.
A detailed analysis shows that if $h_{0} \leq 2w_{0}$ then 
\[
t_{0} = 
\begin{cases}
-x_{0} + w_{0} + 2h_{0} & \text{for $0 < x_{0} \leq h_{0} - w_{0}$}, \\
2w_{0} + h_{0} & \text{for $h_{0} - w_{0} < x_{0} \leq w_{0}$}, \\
x_{0} + w_{0} + h_{0} & \text{for $w_{0} < x_{0} \leq h_{0}$},
\end{cases}
\]
and if $2w_{0} < h_{0}$ then 
\[
t_{0} = 
\begin{cases}
-x_{0} + w_{0} + 2h_{0} & \text{for $0 < x_{0} \leq h_{0}/2$}, \\
x_{0} + w_{0} + h_{0} & \text{for $h_{0}/2 < x_{0} \leq h_{0}$}.
\end{cases}
\]
Substituting $x_{0} = i_{0} - w_{0}$ we obtain 
the upper bound for the statement mentioned in the theorem.

\medskip

\mynewpage

\noindent
(III) The case of the traditional model.

\medskip

In (I), (II) we determined the value of 
${\rm mft}_{\rm bs}(C)$.
This value is $w + 2h$ in case (1), 
is at least $2w + h$ in case (2), and 
is at least $w + (3/2) h$ in case (3) of the statement 
of the theorem.
Therefore ${\rm rad}(C(w, h, i)) = w + h < 
{\rm mft}_{\rm bs}(C(w, h, i))$ and hence 
${\rm mft}_{\rm bs}(C(w, h, i)) = 
{\rm mft}_{\rm tr}(C(w, h, i))$ 
by Lemma \ref{lemma:from_bs_to_tr}.
\end{proof}

\mynewpage

For $\text{\rm g-RECT-WALL}[a, b]$ too, we can derive 
its minimum firing time by slightly modifying 
the derivation in Theorem \ref{theorem:mft_g_rect_wall}.

\medskip

\begin{thm}
\label{theorem:mft_g_rect_wall_a_b}
For any $a$, $b$, 
\[
{\rm mft}_{\text{\rm g-RECT-WALL}[a,b]}(C_{\rm RW}(w, h, i)) = w + h.
\]
\end{thm}

\begin{proof}
By symmetries in the problem 
$\text{\rm g-RECT-WALL}[a,b]$ we may assume that 
$a \leq b$.
The lower bound is obvious because 
$w + h = {\rm rad}(C(w, h, i)) \leq {\rm mft}(C(w$, $h$, $i))$.
We prove the upper bound 
${\rm mft}(C(w_{0}, h_{0}, i_{0})) \leq w_{0} + h_{0}$ 
for a fixed given configuration 
$C(w_{0}, h_{0}, i_{0})$ of $\text{\rm g-RECT-WALL}[a,b]$.
We may assume that $0 \leq i_{0} \leq w_{0} + h_{0}$.
We modify the partial solution $A$ constructed 
in the proof of Theorem \ref{theorem:mft_g_rect_wall} 
as follows.

\medskip

\noindent
(IV) The upper bound for the case of the boundary sensitive model.

\medskip

For the part (I) of the proof of Theorem 
\ref{theorem:mft_g_rect_wall}, 
we change the rule of firing as follows: 
if a node has received at least two of the three messages 
``$x=x_{0}$'', ``$y=y_{0}$'', ``$h=h_{0}$'' before or at time 
$w_{0} + h_{0}$ then it fires at time $w_{0} + h_{0}$.
For the part (II), the rule of firing is: 
if a node has received at least two of the three messages 
``$x=x_{0}$'', ``$y=y_{0}$'', ``$w=w_{0}$'' before or at time 
$w_{0} + h_{0}$ then it fires at time $w_{0} + h_{0}$.

In Figure \ref{figure:fig014} and Figure \ref{figure:fig015} 
the thick dotted lines represent the first time a node 
receives at least two of the three messages.
The lines show that if $C(w, h, i) = C(w_{0}, h_{0}, i_{0})$ 
then all nodes receive two messages before or 
at time $w_{0} + h_{0}$ and hence fire at that time.

Conversely if a node in $C(w, h, i)$ fires at some time 
then at least two of the three statements 
$x = x_{0}$, $y = y_{0}$, $h = h_{0}$ are true for (I) and 
at least two of the three statements 
$x = x_{0}$, $y = y_{0}$, $w = w_{0}$ are true for (II), 
and hence $C(w, h, i) = C(w_{0}, h_{0}, i_{0})$.

Thus we have the upper bound $w_{0} + h_{0}$ for 
the boundary sensitive model.

\medskip

\mynewpage

\noindent
(V) The upper bound for the case of the traditional model.

\medskip

Let $A$ be the partial solution of 
the boundary sensitive model constructed in (IV).
We modify $A$ to construct 
a partial solution $A'$ of the traditional model 
having the same domain and the same firing time.
Note that we cannot use Lemma \ref{lemma:from_bs_to_tr} 
because ${\rm mft}_{\rm bs}(C(w_{0}$, $h_{0}$, $i_{0})) = w_{0} + h_{0} = 
{\rm rad}(C(w_{0}, h_{0}, i_{0}))$. 
First we consider the case $0 < i_{0} < w_{0}$.
There are two modifications.

In the first modification, $A'$ simulates $A$.  
As the general state, 
$A'$ uses the general state of $A$ that corresponds to the 
boundary condition of $v_{\rm gen}$ in $C(w_{0}, h_{0}, i_{0})$.
At the same time, a signal in $A'$ checks the boundary condition 
of $v_{\rm gen}$ from time $0$ to $1$ and 
broadcasts a message ``bc is correct'' at time $1$ 
if the boundary condition is correct.
A node fires at a time $t$ if $A$ fires at $t$ and moreover 
it has received the message ``bc is correct'' before or at $t$.
Then, if $C(w, h, i) = C(w_{0}, h_{0}, i_{0})$ then 
all nodes in $C(w, h, i)$ except the node 
$p_{i_{0} + w_{0} + h_{0}}$ (the opposite node of $p_{i_{0}}$) 
fire at time $w_{0} + h_{0}$ by this firing rule.

The purpose of the second modification is to fire 
$p_{i_{0} + w_{0} + h_{0}}$ at time $w_{0} + h_{0}$.
In the modification, 
two signals ${\rm R}_{2}$, ${\rm R}_{3}$ start at time $0$ 
at $v_{\rm gen}$ and travel through the path, 
${\rm R}_{2}$ travelling clockwise and ${\rm R}_{3}$ 
travelling counterclockwise.
While travelling, ${\rm R}_{2}$ checks the 
statement ``$i = i_{0}$ and $h = h_{0}$'' 
and ${\rm R}_{3}$ checks the statement 
``$w - i = w_{0} - i_{0}$ and $h = h_{0}$.''
If the statement to be checked by the signal $R_{2}$ or $R_{3}$ 
is false the signal vanishes.
A node fires at time $w_{0} + h_{0}$ if 
both of ${\rm R}_{2}$, 
${\rm R}_{3}$ arrive at the node simultaneously 
at the time $w_{0} + h_{0}$.
If $C(w, h, i) = C(w_{0}, h_{0}, i_{0})$ then 
$p_{i_{0}+w_{0}+h_{0}}$ (and only 
$p_{i_{0}+w_{0}+h_{0}}$) fires at time $w_{0} + h_{0}$ 
by this second firing rule.

Conversely suppose that a node in $C(w, h, i)$ fires at some time.
If the node fires by the first firing rule then 
the original $A$ fires the node at that time and 
hence $C(w, h, i) = C(w_{0}, h_{0}, i_{0})$.
If it fires by the second firing rule then 
all of the three statements 
$i = i_{0}$, $w - i = w_{0} - i_{0}$, 
$h = h_{0}$ are true and hence 
$C(w, h, i) = C(w_{0}, h_{0}, i_{0})$.

Thus $A'$ is a partial solution of the traditional model 
that has the domain $\{C(w_{0}, h_{0}, i_{0})\}$ and that fires 
$C(w_{0}, h_{0}, i_{0})$ at time $w_{0} + h_{0}$.

\mynewpage

Next we consider the case $i_{0} = 0$.
In this case ${\rm R}_{2}$ needs one extra step to 
check the statement $i = i_{0}$ (that is, $i = 0$).
However, in this case too we can design the behavior 
of $R_{2}$, $R_{3}$ and specify the second firing rule so that 
\begin{enumerate}
\item[$\bullet$] if $C(w, h, i) = C(w_{0}, h_{0}, 0)$ then 
$p_{w_{0} + h_{0}}$ (and only $p_{w_{0} + h_{0}}$) 
fires at time $w_{0} + h_{0}$ by the second firing rule, and 
\item[$\bullet$] if a node in $C(w, h, i)$ fires at some time 
by the second firing rule then $C(w, h, i) = C(w_{0}, h_{0}, 0)$.
\end{enumerate}

Modifications for other cases 
($- w_{0} - h_{0} < i_{0} < - w_{0}$, $i_{0} = - w_{0}$, 
\ldots, $w_{0} < i_{0} < w_{0} + h_{0}$, $i_{0} = w_{0} + h_{0}$) 
are similar.
\end{proof}

\mynewpage

Although we do not know whether 
$\text{\rm RECT-WALL}$, $\text{\rm g-RECT-WALL}$ have 
minimal-time solutions, at least we have the following 
result.

\begin{thm}
\label{theorem:from_g_rect_wall_to_rect_wall}
If $\text{\rm g-RECT-WALL}$ has a minimal-time solution 
then $\text{\rm RECT-}$ $\text{\rm WALL}$ also has one.
\end{thm}

\begin{proof}
Let $\Gamma$, $\Gamma'$ denote $\text{\rm RECT-WALL}$ and 
$\text{\rm g-RECT-WALL}$ respectively.
If $w \leq h$ then 
${\rm mft}_{\Gamma}(C_{\rm RW}(w, h, 0)) = w + 2h$ 
by Theorem \ref{theorem:mft_rect_wall} 
and 
${\rm mft}_{\Gamma'}(C_{\rm RW}(w,$ \\
$ h, 0)) = w + 2h$ 
by Theorem \ref{theorem:mft_g_rect_wall}.
If $w > h$, then 
${\rm mft}_{\Gamma}(C_{\rm RW}(w, h, 0)) = 2w + h$ 
by Theorem \ref{theorem:mft_rect_wall} 
and 
${\rm mft}_{\Gamma'}(C_{\rm RW}(h, w, w + h)) = 2w + h$ 
by Theorem \ref{theorem:mft_g_rect_wall}, and moreover 
we can easily prove that 
${\rm mft}_{\Gamma'}(C_{\rm RW}(w, h, 0)) 
= {\rm mft}_{\Gamma'}(C_{\rm RW}(h, w, w+h))$.
Therefore $\Gamma'$ is a conservative supervariation of 
$\Gamma$.
\end{proof}

\mynewpage

For $\text{\rm RECT-WALL}[a,b]$ and 
$\text{\rm g-RECT-WALL}[a,b]$ we have the following results 
by Theorems 
\ref{theorem:mft_rect_wall_a_b}, 
\ref{theorem:mft_g_rect_wall_a_b} and 
Yamashita et al's idea.

\begin{thm}
\label{theorem:nonexistence_rect_wall_a_b}
Both of $\text{\rm RECT-WALL}[a,b]$ and 
$\text{\rm g-RECT-WALL}[a,b]$ have no 
minimal-time solutions.
\end{thm}

\begin{proof}
By Theorems \ref{theorem:mft_rect_wall_a_b}, 
\ref{theorem:mft_g_rect_wall_a_b}, 
we have 
${\rm mft}_{\text{\rm RECT-WALL}[a,b]}(C_{\rm RW}(w, h, 0)) 
=$ \\ ${\rm mft}_{\text{\rm g-RECT-WALL}[a,b]}(C_{\rm RW}(w, h, 0)) 
= w + h$.
Therefore $\text{\rm g-RECT-WALL}[a,b]$ is 
a conservative supervariation of $\text{\rm RECT-WALL}[a,b]$.
Hence it suffices to prove that 
$\text{\rm RECT-WALL}[a,b]$ has no minimal-time solutions.

We assume that $\text{\rm RECT-WALL}[a,b]$ has 
a minimal-time solution $\tilde{A}$ and derive a contradiction.
Let $C = C_{\rm RW}(w, h)$ be a configuration of \\
$\text{\rm RECT-WALL}[a,b]$.
$\tilde{A}$ fires $C$ at time ${\rm mft}_{\text{\rm RECT-WALL}[a,b]}(C) 
= w + h$.
As was in the proof of Theorem \ref{theorem:yamashita_result_2}, 
for any $i$ ($-w-h \leq i \leq w+h$) and any $t$ 
($0 \leq t \leq w + h$), by $s_{i}^{t}$ we denote the state 
${\rm st}(p_{i}, t, C, \tilde{A})$.

For any $t$ such that $0 \leq t \leq w+h$, 
$s_{i}^{t} = {\rm Q}$ for $-w-h \leq i < -t$ and 
for $t < i \leq w + h$.  
Using this and assuming that $w$ is sufficiently large, 
we can derive a contradiction by the proof of 
Theorem \ref{theorem:yamashita_result_2} (for $C_{\rm L}(w, h)$) 
without any change.
\end{proof}

\mynewpage

\section{Variations of FSSP for rectangles and squares}
\label{section:rectangles_squares}

In this section we consider the six variations of 
FSSP for rectangles and squares.
They are 
$\text{\rm RECT}$, $\text{\rm RECT}[a,b]$, $\text{\rm SQ}$ 
($= \text{\rm RECT}[1,1]$) 
(see Table \ref{table:fig027}) 
and their generalized 
variations 
$\text{\rm g-RECT}$, $\text{\rm g-RECT}[a,b]$, 
$\text{\rm g-SQ}$ ($= \text{\rm g-RECT}[1,1]$).

A configuration of $\text{\rm RECT}$ is a rectangle 
$C_{\rm R}(w, h)$ (see Figure \ref{figure:fig001} (c)).
We assume that positions in $C_{\rm R}(w, h)$ are 
$(x, y)$ such that $0 \leq x \leq w$, $0 \leq y \leq h$.
When $(r, s)$ is the position of the general in 
a configuration of the generalized variations, 
we denote the configuration by $C_{\rm R}(w, h, (r, s))$.

Minimal-time solutions of the three variations 
$\text{\rm SQ}$, $\text{\rm RECT}$, $\text{\rm g-RECT}$ 
were obtained in the very early days of the research on FSSP.
The minimum firing time of a square $C_{\rm R}(w, w)$ of $\text{\rm SQ}$ 
is $2w$ and a minimal-time solution was obtained in \cite{Shinahr_1974}.
The minimum firing time of a rectangle 
$C_{\rm R}(w, h)$ of $\text{\rm RECT}$ is 
$w + h + \max\{w, h\}$ and minimal-time solutions are shown in 
\cite{Shinahr_1974,Umeo_2007_rectangle_2,Umeo_2007_rectangle_1}.
Finally the minimum firing time of a rectangle 
$C_{\rm R}(w, h, (r, s))$ of $\text{\rm g-RECT}$ is 
$w + h + \max\{w, h\} - r - s$ 
for $0 \leq r \leq w/2$, $0 \leq s \leq h/2$.
Its minimal-time solution was obtained in \cite{Szwerinski_1982}.

$\text{\rm g-SQ}$ is a very natural generalization of 
$\text{\rm SQ}$.
However $\text{\rm g-SQ}$ has not been 
studied widely.
Umeo and Kubo pointed out 
that at present we do not know whether 
$\text{\rm g-SQ}$ has minimal-time solutions or not 
(\cite{Umeo_2012_ACRI}).
The present author derived a formula for the minimum firing time 
of a configuration $C_{\rm R}(w, w, (r, s))$ of 
$\text{\rm g-SQ}$ (\cite{Kobayashi_TCS_2014}).
If $1 \leq w$, $0 \leq r \leq s \leq w/2$, 
the minimum firing time is 
$2w - \min\{r, s - r\}$.
Note that $\min\{r, s - r\}$ is the minimum distance 
from the position $(r, s)$ to the boundary or 
the diagonals of the square.
The derivation of this formula in \cite{Kobayashi_TCS_2014} 
was very tedious using case analysis of $25$ cases.
To simplify the derivation is itself an interesting problem.

At present, for both of 
$\text{\rm RECT}[a,b]$ ($a \not= b$) and 
$\text{\rm g-RECT}[a,b]$ ($a \not= b$), 
we do not know whether they have minimal-time solutions or not.
We derive a formula for the minimum firing time of 
$\text{\rm RECT}[a,b]$ ($a \not= b$).
However, for $\text{\rm g-RECT}[a,b]$ ($a \not= b$) 
we do not know a formula for it.

\mynewpage

\begin{thm}
\label{theorem:mft_rect_a_b}
For any $a$, $b$, 
\[
{\rm mft}_{\text{\rm RECT}[a,b]}(C_{\rm R}(w, h)) = w + h.
\]
\end{thm}

\begin{proof}
The lower bound is obvious because 
$w + h = {\rm rad}(C(w, h)) \leq $ \\
${\rm mft}(C(w, h))$.
For the upper bound we construct 
a check and broadcast partial solution $A$ 
with domain $\{ C(w_{0}, h_{0}) \}$ that fires 
$C(w_{0}, h_{0})$ at time $w_{0} + h_{0}$.
We explain only the main idea.
Suppose that copies of $A$ are placed in 
a configuration $C(w, h)$ of $\text{\rm RECT}[a,b]$.

$A$ checks the statements $w = w_{0}$, $h = h_{0}$.
If they are correct then 
messages ``$w = w_{0}$'', ``$h = h_{0}$'' 
are generated.
A node fires at time $w_{0} + h_{0}$ if it has received at least one 
of the two messages before or at time $w_{0} + h_{0}$.

If $w = w_{0}$ is true the corresponding message is 
generated at $(w_{0}, 0)$ at time $w_{0}$.
If $h = h_{0}$ is true the corresponding message is 
generated at $(0, h_{0})$ at time $h_{0}$.
Therefore, if $C(w, h) = C(w_{0}, h_{0})$ then 
a node $(i, j)$ in $C(w, h)$ receives at least one of 
the two messages at time 
$\min \{ w_{0} + {\rm d}((w_{0}, 0), (i, j)), 
h_{0} + {\rm d}((0, h_{0}), (i, j)) \}$ $=$ 
$\min \{ 2w_{0} - i + j, 2h_{0} + i - j \}$ $=$ 
$\min \{ 2w_{0} - \delta, 2h_{0} + \delta \}$ 
$\leq w_{0} + h_{0}$.
Here $\delta$ denotes $i - j$ and it ranges between 
$-h_{0}$ and $w_{0}$.
Therefore any node in $C(w, h)$ receives at least one of the two 
messages before or at time $w_{0} + h_{0}$.
\end{proof}

\mynewpage

\begin{thm}
\label{theorem:from_g_rect_a_b_to_rect_a_b}
For any $a, b$, 
if $\text{\rm g-RECT}[a,b]$ has a minimal-time solution 
then $\text{\rm RECT}[a,b]$ also has one.
\end{thm}

\begin{proof}
At present we are unable to determine the value 
${\rm mft}_{\Gamma'}(C_{\rm R}(w, h, (i, j)))$ 
for configurations of 
$\Gamma' = \text{\rm g-RECT}[a,b]$.
However at least we can show 
${\rm mft}_{\Gamma'}(C_{\rm R}(w$, $h, (0, 0))) = w + h$ 
as follows.
This implies that $\Gamma'$ is a conservative 
supervariation of $\Gamma = \text{\rm RECT}[a,b]$ 
and we obtain the statement of the theorem.
It is sufficient to prove the upper bound 
${\rm mft}_{\Gamma'}(C(w, h, (0, 0))) \leq w + h$.

Let $A$ be the partial solution (of the boundary sensitive 
mode) of $\Gamma$ constructed in the proof of 
Theorem \ref{theorem:mft_rect_a_b}.
It has the domain $\{ C(w_{0}, h_{0}) \}$ and fires 
$C(w_{0}, h_{0})$ at time $w_{0} + h_{0}$.
We modify $A$ so that at time $0$ it checks that 
the boundary condition of $v_{\rm gen}$ is 
$(1, 1, 0, 0)$ (the boundary condition of the southwest corner) 
and immediately stops if the check fails.
Then we obtain a partial solution $A'$ of 
the boundary sensitive model of $\Gamma'$ that has 
the domain $\{ C(w_{0}, h_{0}, (0, 0)) \}$ 
and that fires $C(w_{0}, h_{0}, (0, 0))$ at time 
$w_{0} + h_{0}$.
Thus we have the upper bound ${\rm mft}(C(w_{0}, h_{0}, (0, 0))) 
\leq w_{0} + h_{0}$ for the boundary sensitive model.

We can change this $A'$ to a partial solution $A''$ of 
the traditional model of $\Gamma'$ with the same domain 
and the same firing time.
The idea for this is similar to that used 
in the proof (V) of Theorem \ref{theorem:mft_g_rect_wall_a_b} 
and we omit it.
\end{proof}

\mynewpage

We can reduce the problem to construct a minimal-time solution 
of $\text{\rm RECT}[a,b]$ to the same problem 
for a variation of L-shaped paths.

For integers $a$ ($\geq 1$), $b$ ($\geq 1$), 
$c$ ($\geq 0$), $d$ ($\geq 0$), 
we define $\text{\rm LSP-C}[a, b; c, d]$ to be the variation 
such that configurations are 
generalized L-shaped paths 
$C_{\rm L}(w, h, w)$ such that 
$w = al + c$, $h = bl + d$ for some 
$l$ ($\geq 0$).
Note that the general $p_{w}$ is 
at the southeast corner node.
``${\rm C}$'' in ``$\text{\rm LSP-C}$'' means that 
``the general is at the corner.''
In Figure \ref{figure:fig032} we show a configuration 
$C_{\rm L}(13, 10, 13)$ of 
$\text{\rm LSP-C}[3, 2; 1, 2]$.
In this case, $l = 4$, $w = 3 \cdot 4 + 1 = 13$, 
$h = 2 \cdot 4 + 2 = 10$.

\begin{figure}[htb]
\begin{center}
\includegraphics[scale=1.0]{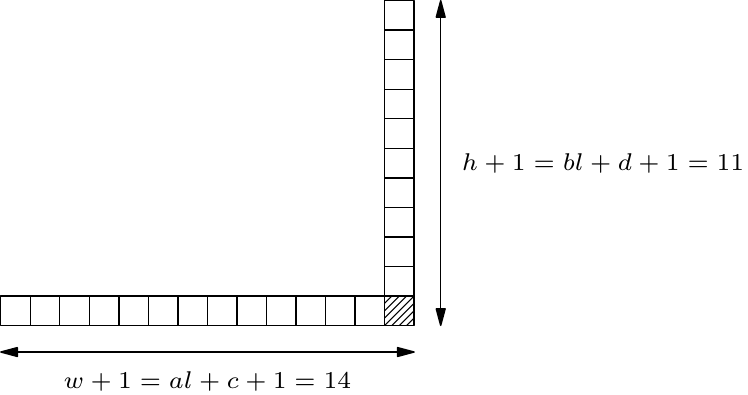}
\end{center}
\caption{A configuration $C_{\rm L}(13, 10, 13)$ of 
$\text{\rm LSP-C}[3, 2; 1, 2]$.}
\label{figure:fig032}
\end{figure}

\mynewpage

\begin{thm}
\label{theorem:mft_lsp_c_a_b_c_d}
For any $a, b, c, d$, 
\[
{\rm mft}_{\text{\rm LSP-C}[a,b;c,d]}(C_{\rm L}(w, h, w)) 
= w + h.
\]
\end{thm}

\begin{proof}
We explain only the main idea for constructing 
a partial solution $A$ of 
the traditional model of 
$\Gamma = \text{\rm LSP-C}[a,b;c,d]$ 
for one fixed configuration 
$C(w_{0}, h_{0}, w_{0})$.
Suppose that copies of $A$ 
are placed on a configuration $C(w, h, w)$ of $\Gamma$.

$A$ checks the two conditions $w = w_{0}$, $h = h_{0}$ 
and broadcasts messages ``$w = w_{0}$'', ``$h = h_{0}$'' 
if conditions are correct.
A node fires at time $w_{0} + h_{0}$ if it has received 
at least one of the two messages before or at 
time $w_{0} + h_{0}$.

If $C(w, h, w) = C(w_{0}, h_{0}, w_{0})$, 
the message ``$w = w_{0}$'' is generated at $p_{0}$ at time 
$w_{0}$ and the message ``$h = h_{0}$'' is generated at 
$p_{w_{0} + h_{0}}$ at time $h_{0}$.
Therefore a node $p_{i}$ receives at least one of 
the two messages before or at time 
$\min \{ w_{0} + i, h_{0} + (w_{0} + h_{0} - i) \} 
\leq w_{0} + h_{0}$.
By this we have the upper bound 
${\rm mft}(C(w_{0}, h_{0}, w_{0})) \leq w_{0} + h_{0}$.
Moreover, at time $w_{0} + h_{0} - 1$ the node $p_{h_{0}}$ 
receives none of the two messages.
Using this and Lemma \ref{lemma:fact_4} 
we can prove the lower bound $w_{0} + h_{0} \leq 
{\rm mft}(C(w_{0}, h_{0}, w_{0}))$.
\end{proof}

\mynewpage

\begin{thm}
\label{theorem:from_lsp_c_a_b_c_d_to_rect_a_b}
Let $a$, $b$ be any numbers.
If $\text{\rm LSP-C}[a,b;c,d]$ has a minimal-time solution 
for any $c, d$ then 
$\text{\rm RECT}[a,b]$ has a minimal-time solution.
\end{thm}

\begin{proof}
We use the idea used in \cite{Shinahr_1974} 
to construct a minimal-time solution of $\text{\rm RECT}$.
The idea is to cover a configuration $C_{\rm R}(w, h)$ 
of $\Gamma = \text{\rm RECT}[a,b]$ 
with 
configurations of $\Gamma' = \text{\rm LSP-C}[a,b;c,d]$ 
(with various values of $c, d$) 
that are modified so that the two directions the west and 
the east are exchanged (and hence a configuration 
has a southwest corner instead of a southeast corner).

We explain the idea using 
the configuration $C_{\rm R}(10, 6)$ of 
$\text{\rm RECT}[5, 3]$.
We cover this configuration with $11$ 
configurations of $\text{\rm LSP-C}[5, 3; c, d]$ 
as shown in Figure \ref{figure:fig054}.

\begin{figure}[htb]
\begin{center}
\includegraphics[scale=1.0]{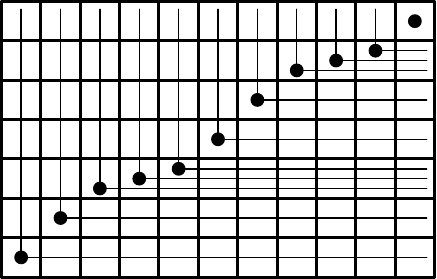}
\end{center}
\caption{Eleven configurations of 
$\text{\rm LSP-C}[5, 3; c, d]$ that cover 
the configuration $C_{\rm R}(10, 6)$ 
of $\text{\rm RECT}[5, 3]$.}
\label{figure:fig054}
\end{figure}

In this figure an L-shaped line represents 
a configuration of $\Gamma'$ 
and the black circle on it represents the position 
of the general.
We use the supposed minimal-time solution $A$ of 
$\text{\rm LSP-C}[a,b;c,d]$ to fire the nodes 
in the configuration.
For $A$ too, we modify it so that the two directions 
the west and the east are exchanged.

In Table \ref{table:fig055} we show basic data of 
the $11$ configurations of $\Gamma'$.
The column $v_{\rm gen}$ represents the position of the general, 
the column ``activation time'' represents the time when 
the firing of the configuration is triggered, and 
the column ``firing time'' represents the actual time 
when nodes in the configuration fire.
If $v_{\rm gen}$ is at $(i, j)$ then the activation time 
is $i + j$, the configuration is 
$C_{\rm L}(w - i, h - j, w - i)$, 
and hence all nodes in it fire at 
$(i + j) + \{ (w - i) + (h - j) \} = w + h$, the minimum 
firing time of $C_{\rm R}(w, h)$ of $\text{\rm RECT}[a,b]$.

\begin{table}[htb]
\begin{center}
\includegraphics[scale=1.0]{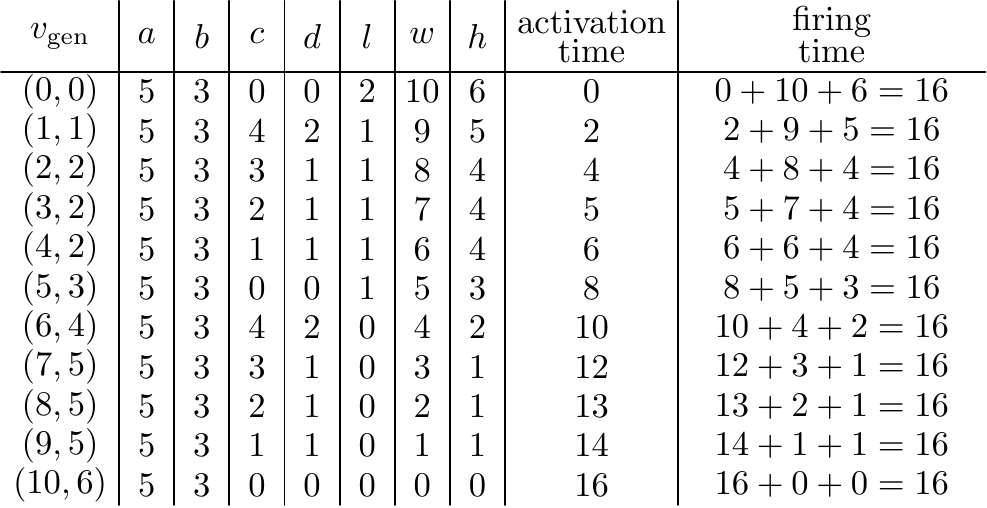}
\end{center}
\caption{The basic data on the eleven configurations 
shown in Figure \ref{figure:fig054}.}
\label{table:fig055}
\end{table}

\end{proof}

\mynewpage

\section{Some applications of Yamashita et al's idea}
\label{section:applications}

We have proved nonexistence of minimal-time solutions 
for some variations using Yamashita et al's idea.
For these results the application of the idea was 
straightforward and needed no modifications.
In this section we show two examples of application of the idea 
where slight modifications of the idea are necessary.

\begin{ex}
\label{example:ex001}
\end{ex}

Each of the variations for which we proved nonexistence 
of minimal-time solutions using Yamashita et al's idea 
contained an infinitely many configurations $C$ 
such that ${\rm rad}(C) = {\rm mft}(C)$.
Hence, these results might give an impression 
that existence of such configurations 
is necessary for Yamashita et al's idea to work.
In this example, we disprove this impression 
by proving nonexistence for a variation 
$\Gamma_{\rm ex1}$ 
in which ${\rm rad}(C) < {\rm mft}(C)$ for any 
configuration $C$.
This variation is very artificial and 
we include it only for the above mentioned purpose.

A configuration of 
$\Gamma_{\rm ex1}$ 
is an L-shaped path $C_{\rm L}(w, h)$ for which 
there is an integer $r$ ($\geq 1$) such that 
$w = 2^{r^{2}}$, 
$h = 2^{r^{2}} - r = w - \sqrt{\log w}$.
For this variation we have 
${\rm mft}(C(w, h)) = 2w > w + h = {\rm rad}(C(w, h))$.

First we prove ${\rm mft}(C(w, h)) = 2w$.
The lower bound $2w \leq {\rm mft}(C(w, h))$ can be proved 
by Lemma \ref{lemma:fact_4} using the intuitive idea 
that at time $2w - 1$ the node $p_{0}$ cannot know 
the values of $w, h$.

The upper bound ${\rm mft}(C(w, h)) \leq 2w$ can be proved 
by constructing a partial solution $A$ with the domain 
$\{ C(w_{0}, h_{0}) \}$ that fires 
$C(w_{0}, h_{0})$ at time $2w_{0}$ for each fixed 
configuration $C(w_{0}, h_{0})$ of $\Gamma_{\rm ex1}$.
$A$ uses the message ``$w = w_{0}$'' that is generated at 
$p_{w_{0}}$ at time $w_{0}$. 
We omit the detail of the construction.
(Note that in this variation too, one of the two values $w, h$ 
in $C(w, h)$ completely determines the value of the other.)

Nonexistence of minimal-time solutions of $\Gamma_{\rm ex1}$ 
can be proved by modifying the proof of 
Theorem \ref{theorem:yamashita_result_2} as follows.

If $w$ is sufficiently large, 
there exist times $t_{0}$, $t_{1}$ such that 
$w + 2 + 2r \leq t_{0} < t_{1} \leq w + h - 1$ and 
$s_{t_{0}-2r-1}^{t_{0}} \ldots s_{t_{0}}^{t_{0}} 
= s_{t_{1}-2r-1}^{t_{1}} \ldots s_{t_{1}}^{t_{1}}$.
This follows from the following facts.
The number of time $t$ such that 
$w + 2 + 2r \leq t \leq w + h - 1$ is 
$2^{r^{2}} - 3r - 2$.
The number of sequences $s_{-2r-1} \ldots s_{0}$ of length $2r + 2$ of 
states of the finite automaton $A$ is 
$u^{2r + 2}$.
Here $A$ is the minimal-time solution of 
$\Gamma_{\rm ex1}$ which we assume to 
exist and $u$ is the number of states of $A$.
However we have 
\[
2^{r^{2}} - 3r - 2 > u^{2r + 2}
\]
for all sufficiently large $r$.

By the same inference as that used in the proof of 
Theorem \ref{theorem:yamashita_result_2} 
(Figure \ref{figure:fig005} (a)), 
we have 
$s_{t_{2}-2r-1}^{t_{2}} \ldots s_{t_{2}}^{t_{2}} 
= s_{w+h-2r-2}^{w+h-1} \ldots s_{w+h-1}^{w+h-1}$, 
where 
$t_{2} = w + h - 1 + t_{0} - t_{1}$.
From this we have 
$s_{t_{2}-2r}^{t_{2}+1} \ldots s_{t_{2}}^{t_{2}+1} 
= s_{w+h-2r-1}^{w+h} \ldots s_{w+h-1}^{w+h}$. 
This is because both of 
$s_{t_{2} + 1}^{t_{2}}$, $s_{w + h}^{w + h - 1}$ 
are ${\rm Q}$ (Figure \ref{figure:fig005} (b)).
Starting with this, we have 
$s_{t_{2}-2r+1}^{t_{2}+2} \ldots s_{t_{2}-1}^{t_{2}+2} 
= s_{w+h-2r}^{w+h+1} \ldots s_{w+h-2}^{w+h+1}$, 
$s_{t_{2}-2r+2}^{t_{2}+3} \ldots s_{t_{2}-2}^{t_{2}+3} 
= s_{w+h-2r+1}^{w+h+2} \ldots$ $s_{w+h-3}^{w+h+2}$, 
\ldots, and finally 
$s_{t_{2}-r}^{t_{2}+r+1} = s_{2h-1}^{2w}$.
But $s_{2h-1}^{2w}$ is the firing state ${\rm F}$ 
because $A$ fires $C(w, h)$ at time $2w$.
This means that  $s_{t_{2}-r}^{t_{2}+r+1}$ is also 
${\rm F}$ and this is a contradiction 
because $t_{2} + r + 1 = 2w - (t_{1} - t_{0}) < 2w$.

\hfill (End of Example \ref{example:ex001})

\mynewpage

\begin{ex}
\label{example:ex003}
\end{ex}

We define four variations 
$\Gamma_{\rm ex2a}$, 
$\Gamma_{\rm ex2b}$, 
$\Gamma_{\rm ex2c}$, 
$\Gamma_{\rm ex2d}$
of FSSP.
A configuration of these variations is 
obtained from a square $C_{\text{\rm R}}(w, w)$ 
by deleting some nodes. 
Figure \ref{figure:fig020} (a), (b), (c), (d) respectively 
show forms of configurations of 
$\Gamma_{\rm ex2a}$, 
$\Gamma_{\rm ex2b}$, 
$\Gamma_{\rm ex2c}$, 
$\Gamma_{\rm ex2d}$
respectively.
A configuration is a square of width and height $w$ 
and has $w/2$ L-shaped slits or $w^{2} / 4$ holes.  
$w$ is an even number and $w = 10$ in Figure \ref{figure:fig020}.
Let $C(w)$ be such a configuration of one of the four variations.
It contains the rectangular wall $C_{\rm RW}(w, w)$.
We continue to use $p_{-2w}, \ldots, p_{0}, \ldots, p_{2w}$ 
to denote nodes in the part $C_{\rm RW}(w, w)$ of $C(w)$.

\begin{figure}[htb]
\begin{center}
\includegraphics[scale=1.0]{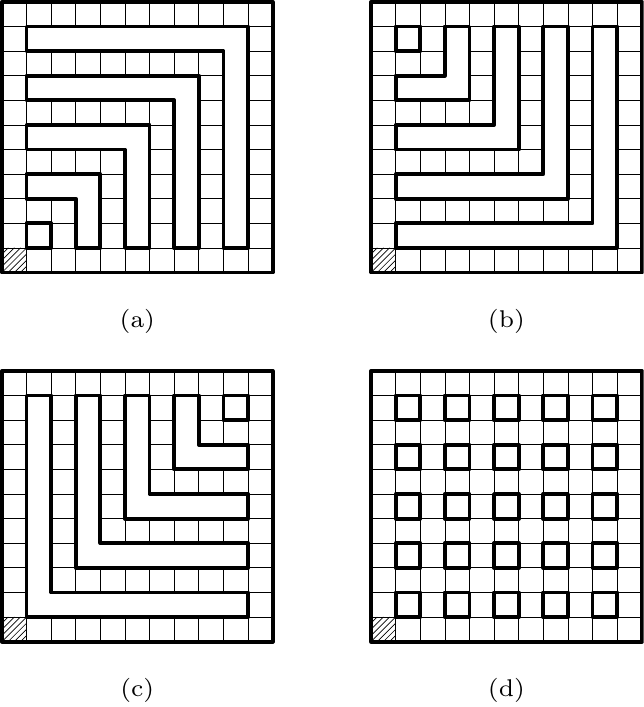}
\end{center}
\caption{(a), (b), (c), (d) respectively represent the forms of 
configurations of 
$\Gamma_{\rm ex2a}$, 
$\Gamma_{\rm ex2b}$, 
$\Gamma_{\rm ex2c}$, 
$\Gamma_{\rm ex2d}$ 
respectively.}
\label{figure:fig020}
\end{figure}

We can show that ${\rm mft}(C(w)) = 2w = {\rm rad}(C(w))$ 
for each of the four variations.
For 
$\Gamma_{\rm ex2a}$, 
$\Gamma_{\rm ex2c}$ and 
$\Gamma_{\rm ex2d}$ 
we use two messages ``$w = w_{0}$'' generated at 
$(0, w_{0})$ and  $(w_{0}, 0)$ at time $w_{0}$.
For 
$\Gamma_{\rm ex2b}$ 
we use messages ``$w = w_{0}$'' 
that are generated at corners 
$(w_{0}, 0)$, $(w_{0} - 2, 2)$, $(w_{0} - 4, 4)$, \ldots, 
$(2, w_{0} - 2)$, $(0, w_{0})$ at time $w_{0}$.

We prove that 
$\Gamma_{\rm ex2a}$, 
$\Gamma_{\rm ex2b}$, 
$\Gamma_{\rm ex2c}$ 
have no minimal-time solutions and 
$\Gamma_{\rm ex2d}$ 
has a minimal-time solution.

We can prove nonexistence of minimal-time solutions 
of $\Gamma_{\rm ex2a}$ and 
$\Gamma_{\rm ex2b}$ 
using the proof of Theorem \ref{theorem:nonexistence_rect_wall_a_b} 
for $\text{\rm RECT-WALL}[1,1]$ 
without any modifications.
For 
$\Gamma_{\rm ex2c}$ 
the following modification is necessary. 
This is because in this variation there are entrances to 
side corridors in both of the north and the east walls.

For an integer $i$, let $q_{i}$ denote $p_{-i}$.
Then $q_{0}, q_{1}, \ldots, q_{2w}$ is 
the sequence of nodes obtained by going from 
the southwest corner (the position of the general) 
to the northwest corner, and then to the northeast corner.
For any $i$ ($0 \leq i \leq 2w$) and any $t$, 
by $s_{i}^{t}$ we denote the state 
${\rm st}(q_{i}, t, C(w), \tilde{A})$.
Here $\tilde{A}$ is the minimal-time solution of 
$\Gamma_{\rm ex3c}$ which we assume to exist.

If $w$ is sufficiently large there are 
two even times $t_{0}$, $t_{1}$ such that 
$w + 2 \leq t_{0} < t_{1} \leq 2w - 2$ 
and $s_{t_{0}-1}^{t_{0}} s_{t_{0}}^{t_{0}} 
= s_{t_{1}-1}^{t_{1}} s_{t_{1}}^{t_{1}}$.
We can show that if $t_{1} \leq 2w - 4$ then 
$s_{t_{0}+1}^{t_{0}+2} s_{t_{0}+2}^{t_{0}+2} 
= s_{t_{1}+1}^{t_{1}+2} s_{t_{1}+2}^{t_{1}+2}$.
In Figure \ref{figure:fig021} (a) we explain this inference 
in the same way as we explained similar inferences 
in Figure \ref{figure:fig005} for the proof of 
Theorem \ref{theorem:yamashita_result_2}.

\begin{figure}[p]
\begin{center}
\includegraphics[scale=1.0]{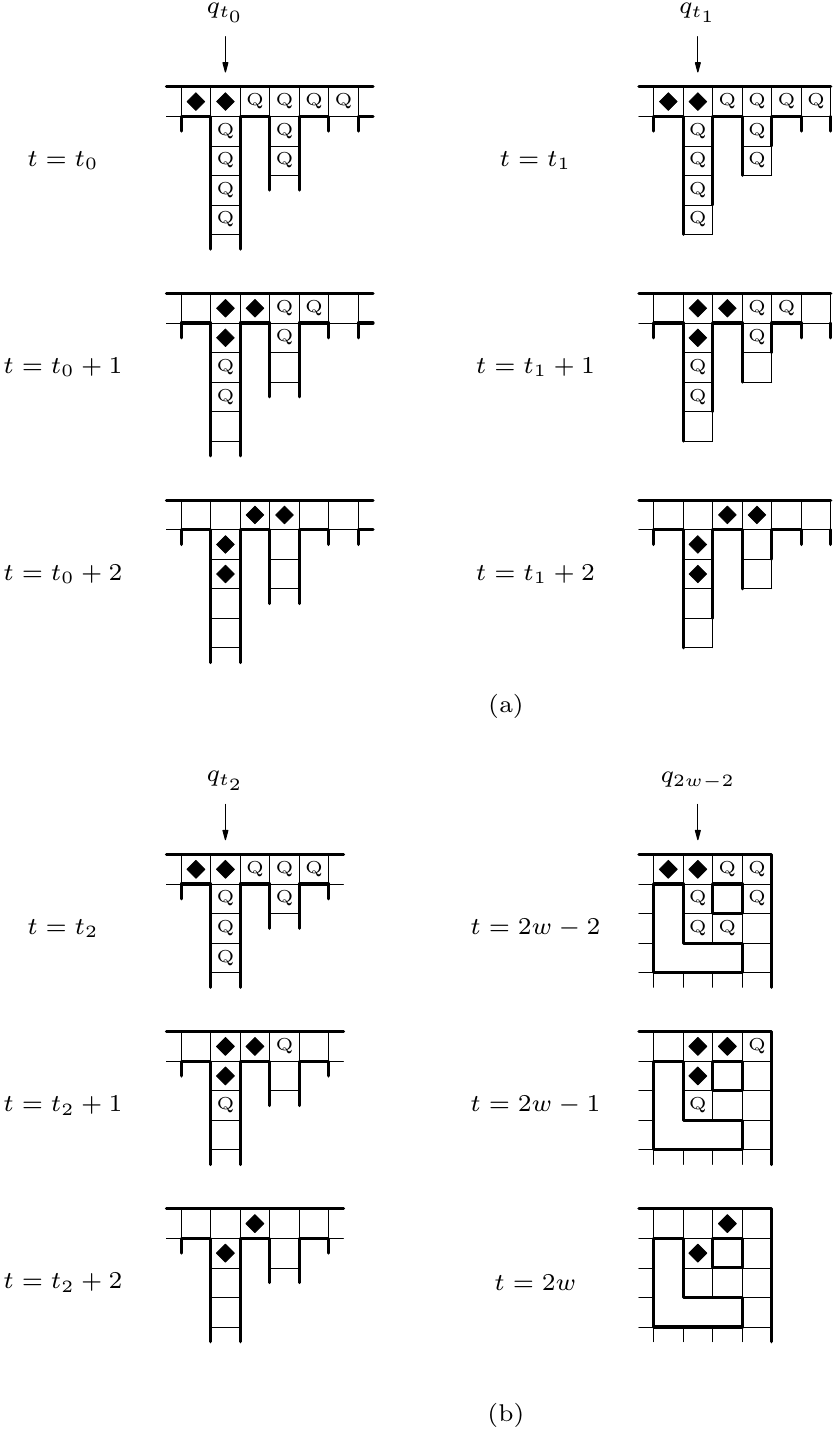}
\end{center}
\caption{(a) A diagram that shows why 
$s_{t_{0}-1}^{t_{0}} s_{t_{0}}^{t_{0}}
 = s_{t_{1}-1}^{t_{1}} s_{t_{1}}^{t_{1}}$ 
implies 
$s_{t_{0}+1}^{t_{0}+2} s_{t_{0}+2}^{t_{0}+2}
 = s_{t_{1}+1}^{t_{1}+2} s_{t_{1}+2}^{t_{1}+2}$.
(b) Comparison of changes of states of nodes near $q_{t_{2}}$ and 
that of nodes near $q_{2w - 2}$.}
\label{figure:fig021}
\end{figure}

Repeating this inference we can show that 
$s_{t_{2}-1}^{t_{2}} s_{t_{2}}^{t_{2}} = 
s_{2w-3}^{2w-2} s_{2w-2}^{2w-2}$, 
where $t_{2}$ is the even time 
$t_{2} = 2w - 2 - t_{1} + t_{0}$.
In Figure \ref{figure:fig021} (b) 
we show the change of states of nodes 
near $q_{t_{2}}$ from time $t_{2}$ to $t_{2}+2$ 
and that of nodes near $q_{2w-2}$ from time 
$2w - 2$ to $2w$.

The lower part of this figure contains two pairs of 
diamonds and one of the pairs means  
$s_{t_{2} + 1}^{t_{2} + 2} = s_{2w - 1}^{2w}$.
However $s_{2w - 1}^{2w} = {\rm F}$ because 
$\tilde{A}$ fires all nodes in $C(w)$ at time $2w$.
Therefore $s_{t_{2} + 1}^{t_{2} + 2} = {\rm F}$ 
and we have a contradiction because 
$t_{2} + 2 < 2w$.
This completes the proof of nonexistence for 
$\Gamma_{\rm ex2c}$.

\mynewpage

The idea for constructing a solution of 
$\Gamma_{\rm ex2d}$ that fires nodes in $C(w)$ at time 
${\rm mft}_{\rm ex2d}(C(w)) = 2w$ is as follows 
(see Figure \ref{figure:fig057}).

\begin{figure}[htb]
\begin{center}
\includegraphics[scale=1.0]{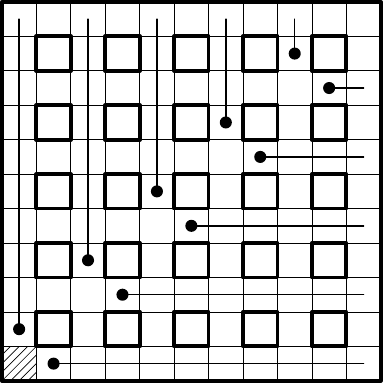}
\end{center}
\caption{The idea for constructing a minimal-time solution 
of $\Gamma_{\rm ex2d}$.}
\label{figure:fig057}
\end{figure}

For each $i$ ($0 \leq i \leq w/2$) a signal 
can locate the position $(2i, 2i)$ on the 
diagonal at time $4i$.
Therefore, for each $i$ ($0 \leq i \leq w/2 - 1$) 
signals can activate the FSSP of 
the following two lines with $w - 2i$ nodes 
at time $4i + 1$.
One is the horizontal line of 
positions $(2i+1, 2i)$, $(2i+2, 2i)$, \ldots, $(w, 2i)$.
The other is the vertical line of 
positions $(2i, 2i+1)$, $(2i, 2i+2)$, \ldots, $(2i, w)$.
All nodes in these two lines fire at time 
$(4i + 1) + \{2 (w - 2i) - 2 \} = 2w - 1$.
Moreover any node in $C(w)$ except $(w, w)$ that is not 
in these lines is adjacent to a node in the lines.
Using this idea we can construct a finite automaton 
that fires any node except $(w, w)$ at time $2w$.
To modify it so that it also fires $(w, w)$ at time $2w$ 
is easy.
\hfill (End of Example \ref{example:ex003})

\mynewpage

\section{Conclusion}
\label{section:conclusion}

We consider variations of FSSP only with respect to 
shapes of configurations (that is, variations 
that satisfy the four conditions mentioned in 
Section \ref{section:introduction}) and try to prove 
nonexistence of minimal-time solutions for them.

Concerning this problem we have some results 
that require assumptions from complexity theory 
(\cite{Goldstein_Kobayashi_SIAM_2005,
Goldstein_Kobayashi_SIAM_2012,
Kobayashi_TCS_2001}).
The recent result by Yamashita et al (\cite{Yamashita_et_al_2014}) 
restated as a result on $\text{\rm LSP}[a,b]$ 
is the first nonexistence proof without any assumption 
(as far as the author knows).
In the present paper, we proved nonexistence of 
minimal-time solutions using their idea 
for the three variations 
$\text{\rm g-LSP}[a,b]$, 
$\text{\rm RECT-WALL}[a,b]$, 
$\text{\rm g-RECT-WALL}[a,b]$.

Their idea can be used only for very special types of variations.
However, at present it is the only one tool we have 
for proving nonexistence without any assumption.
It is very challenging to extend their idea and 
to find other such ideas.

\mynewpage

We list open problems.

\medskip

\noindent
{\bf Problem 1}. For each of the following five 
variations, determine whether it has minimal-time 
solutions or not:
$\text{\rm RECT-WALL}$, 
$\text{\rm g-RECT-WALL}$, 
$\text{\rm RECT}[a, b]$ 
(excluding $\text{\rm SQ} = \text{\rm RECT}[1,1]$), 
$\text{\rm g-RECT}[a, b]$ 
(including $\text{\rm g-SQ} = \text{\rm g-RECT}[1,1]$), 
$\text{\rm LSP-C}[a,b;c,d]$.

\begin{figure}[htb]
\begin{center}
\includegraphics[scale=1.0]{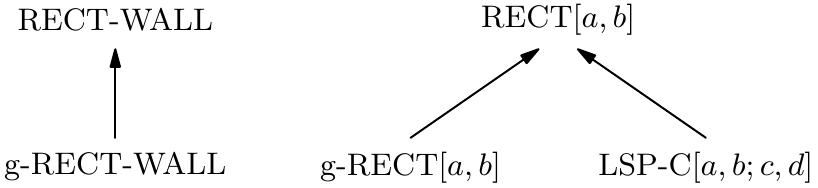}
\end{center}
\caption{Relations between the five variations of FSSP.}
\label{figure:fig056}
\end{figure}

In Figure \ref{figure:fig056} we show some relations between 
these five variations.
An arrow from $\Gamma$ to $\Gamma'$ means that 
if $\Gamma$ has a minimal-time solution then $\Gamma'$ 
also has one.

It seems that we cannot prove nonexistence for these 
five variations with Yamashita et al's idea.
Configurations of $\text{\rm RECT}[a,b]$, 
$\text{\rm g-RECT}[a,b]$ have no corridors 
such as in configurations $C_{\rm L}(w, h)$, 
$C_{\rm RW}(w, h)$ and 
$C(w)$ of $\Gamma_{\rm ex2a}$, 
$\Gamma_{\rm ex2b}$, 
$\Gamma_{\rm ex2c}$. 
For the three variations $\text{\rm RECT-WALL}$, 
$\text{\rm g-RECT-WALL}$, 
$\text{\rm LSP-C}[a,b;c,d]$ we have 
${\rm rad}(C) < {\rm mft}(C)$ and 
moreover ${\rm mft}(C) - {\rm rad}(C)$ is considerably large 
(see our comment in Example 1).

Therefore, proving nonexistence for one of these five variations 
implies finding another new idea for proving nonexistence 
without any assumption, and is really challenging.

\medskip

\noindent
{\bf Problem 2}. To determine the minimum firing time of $\text{\rm g-RECT}[a,b]$.
For the special case $\text{\rm g-SQ} = 
\text{\rm g-RECT}[1,1]$ we determined the value 
(\cite{Kobayashi_TCS_2014}).  
Its derivation was very complicated and tedious. 
However the derived value itself is very simple.
(See the survey at the beginning of 
Section \ref{section:rectangles_squares}.)
It is interesting to derive the value for the 
case $a = 1, b = 2$ and moreover for general cases 
of $a, b$.

\medskip

\noindent
{\bf Problem 3}. 
To each of the four variations of FSSP of rectangular walls 
(that is, $\text{\rm RECT-WALL}$, $\text{\rm RECT-WALL}[a,b]$, 
$\text{\rm g-RECT-WALL}$, $\text{\rm g-RECT-WALL}[a,b]$)
there naturally corresponds its ``one-way'' variation.
In it we require the information flow to be clockwise.
(Formally, we use eight types of finite automata, 
each corresponding to how information flows, 
from the south to the north, 
from the south to the east, 
from the west to the east, 
from west to the south, ..., 
from the east to the west, 
and from the east to the north.
We construct rectangular walls from copies of these 
finite automata so that information flows clockwise.)
We know that such variations have solutions 
because we already have (minimal-time) solutions for 
FSSP of one-way rings 
(see the survey at the beginning of 
Section \ref{section:rec_walls}).
The open problems are to determine their 
minimum firing times and to know whether they have 
minimal-time solutions or not.

\mynewpage

\bibliographystyle{plain}
\bibliography{ms}

\end{document}